\documentclass[12pt]{amsart}
\usepackage{amssymb}
\usepackage{color}
\usepackage{amsmath,epic,curves,amscd}
\usepackage[english]{babel}
\usepackage{graphicx}
\usepackage{comment}
\usepackage{appendix}
\usepackage{mathdots}
\usepackage{multirow}
\usepackage[all]{xy}
\newtheorem{claim}{}[section]
\newtheorem{theorem}[claim]{Theorem}

\newtheorem{lemma}[claim]{Lemma}
\newtheorem{proposition}[claim]{Proposition}
\newtheorem{corollary}[claim]{Corollary}

\theoremstyle{remark}

\renewenvironment{proof}{\noindent{\it Proof. \hskip0pt}}
                      {$\square$\par\medskip}

\begin{document}
\baselineskip 6.0 truemm
\parindent 1.5 true pc

\newcommand\xx{{\text{\sf X}}}
\newcommand\lan{\langle}
\newcommand\ran{\rangle}
\newcommand\tr{{\text{\rm Tr}}\,}
\newcommand\ot{\otimes}
\newcommand\ol{\overline}
\newcommand\join{\vee}
\newcommand\meet{\wedge}
\renewcommand\ker{{\text{\rm Ker}}\,}
\newcommand\image{{\text{\rm Im}}\,}
\newcommand\id{{\text{\rm id}}}
\newcommand\tp{{\text{\rm tp}}}
\newcommand\pr{\prime}
\newcommand\e{\epsilon}
\newcommand\la{\lambda}
\newcommand\inte{{\text{\rm int}}\,}
\newcommand\ttt{{\text{\rm t}}}
\newcommand\spa{{\text{\rm span}}\,}
\newcommand\conv{{\text{\rm conv}}\,}
\newcommand\rank{\ {\text{\rm rank of}}\ }
\newcommand\re{{\text{\rm Re}}\,}
\newcommand\ppt{\mathbb T}
\newcommand\rk{{\text{\rm rank}}\,}
\newcommand\HA{{\mathcal H}_A}
\newcommand\HB{{\mathcal H}_B}
\newcommand\HC{{\mathcal H}_C}
\newcommand\CI{{\mathcal I}}
\newcommand{\bra}[1]{\langle{#1}|}
\newcommand{\ket}[1]{|{#1}\rangle}
\newcommand\cl{\mathcal}
\newcommand\idd{{\text{\rm id}}}
\newcommand\OMAX{{\text{\rm OMAX}}}
\newcommand\OMIN{{\text{\rm OMIN}}}
\newcommand\diag{{\text{\rm Diag}}\,}
\newcommand\calI{{\mathcal I}}
\newcommand\bfi{{\bf i}}
\newcommand\bfj{{\bf j}}
\newcommand\bfk{{\bf k}}
\newcommand\bfl{{\bf l}}
\newcommand\bfp{{\bf p}}
\newcommand\bfq{{\bf q}}
\newcommand\bfzero{{\bf 0}}
\newcommand\bfone{{\bf 1}}
\newcommand\sa{{\rm sa}}
\newcommand\ph{{\rm ph}}
\newcommand\phase{{\rm ph}}
\newcommand\res{{\text{\rm res}}}
\newcommand{\algname}[1]{{\sc #1}}
\newcommand{\Setminus}{\setminus\hskip-0.2truecm\setminus\,}
\newcommand\calv{{\mathcal V}}
\newcommand\calg{{\mathcal G}}
\newcommand\calt{{\mathcal T}}
\newcommand\calvnR{{\mathcal V}_n^{\mathbb R}}
\newcommand\D{{\mathcal D}}
\newcommand\C{{\mathcal C}}
\newcommand\aaa{{\mathcal A}}
\newcommand\bbb{{\mathcal B}}
\newcommand\ccc{{\mathcal C}}
\newcommand\xxxx{\par\bigskip {\color{red}========================================} \bigskip\par}
\newcommand\tefrac{\textstyle\frac}
\newcommand\xxx{\xx^\sigma}
\newcommand{\parallelsum}{\mathbin{\!/\mkern-5mu/\!}}
\newcommand\ext{{\rm Ext}\,}
\newcommand\E{{\mathcal E}}
\newcommand\W{{\mathcal W}}
\newcommand\X{{\mathcal X}}
\newcommand\variable{{\mathcal X}}

\title{On the convex cones arising from classifications of partial entanglement in the three qubit system}

\author{Kyung Hoon Han and Seung-Hyeok Kye}
\address{Kyung Hoon Han, Department of Data Science, The University of Suwon, Gyeonggi-do 445-743, Korea}
\email{kyunghoon.han at gmail.com}
\address{Seung-Hyeok Kye, Department of Mathematics and Institute of Mathematics, Seoul National University, Seoul 151-742, Korea}
\email{kye at snu.ac.kr}
\thanks{Both KHH and SHK were partially supported by NRF-2017R1A2B4006655, Korea}

\subjclass{81P40, 52A20, 15A69}

\keywords{}

\begin{abstract}
In order to classify partial entanglement of multi-partite states, it is
natural to consider the convex hulls, intersections and differences
of basic convex cones obtained from partially separable states with respect
to partitions of systems. In this paper, we consider convex cones
consisting of {\sf X}-shaped three qubit states arising in this way.
The class of {\sf X}-shaped states includes important classes like
Greenberger-Horne-Zeilinger diagonal states. We find all the extreme
rays of those convex cones to exhibit corresponding partially
separable states. We also give characterizations for those cones
which give rise to necessary criteria in terms of diagonal and
anti-diagonal entries for general three qubit states.
\end{abstract}

\maketitle

\section{Introduction}

The notion of entanglement is now considered as an indispensable
resource in the current quantum information theory. In the
multi-partite systems, there are various notions of separability
according to partitions of systems, which give rise to different
kinds of partial entanglement. In the tri-partite system, we
may consider three kinds of bi-partitions $A$-$BC$, $B$-$CA$ and
$C$-$AB$ of systems. In this way, a tri-partite state may be
considered as a bi-partite state with respect to one of the above
bi-partitions. It was shown in \cite{bdmsst} that a three qubit
state may be entangled even though it is separable as a bi-partite
state with respect to any bi-partitions. Therefore, it is natural to
classify partial separability in multi-partite systems, as they were
suggested in the liturature
\cite{{acin},{coffman},{dur-cirac},{dur-cirac-tarrach},{dur-vidal},{seevinck-uffink},{sz2011},{sz2015},{sz2012}}.
We recall that a multi-partite state is (fully) separable if it is a
convex sum of pure product states, and entangled if it is not
separable.

We will work in the real vector space of all three qubit self-adjoint matrices,
and consider the convex cones $\aaa$, $\bbb$ and $\ccc$ consisting of all unnormalized
separable states with respect to the bi-partitions $A$-$BC$, $B$-$CA$ and $C$-$AB$, respectively.
Recall that a subset $C$ of a real vector space is called a {\sl convex cone} when
$C+C\subset C$ and $aC\subset C$ for $a\ge 0$. 
We note that the sum $C_1+C_2$ of two convex cones $C_1$ and $C_2$ is again a convex cone
which coincides with the convex hull of $C_1$ and $C_2$, that is, the smallest convex set containing $C_1$ and $C_2$.

The above mentioned result \cite{bdmsst} tells us that the convex cone
$\aaa\cap\bbb\cap\ccc$ is strictly bigger than the convex cone of all
fully separable states as tri-partite states. The differences $\aaa\setminus (\bbb\cup\ccc)$
and $(\aaa\cap\bbb)\setminus \ccc$ have been considered in
\cite{{dur-cirac},{dur-cirac-tarrach}}, together with similar sets
obtained by permuting $\aaa$, $\bbb$ and $\ccc$. On the other
hand, the convex hull $\aaa+\bbb+\ccc$
and the difference $(\aaa+\bbb+\ccc)\setminus
(\aaa\cup\bbb\cup\ccc)$ have been also considered in \cite{acin} and
\cite{seevinck-uffink}, respectively.
More recently, all the possible classes
\begin{equation}\label{partition}
[C_1\cap\cdots \cap C_k]\setminus [C_{k+1}\cup\cdots \cup C_\ell]
\end{equation}
have been considered in \cite{sz2012}, where $C_i$ is one of the following convex cones
\begin{equation}\label{class}
\aaa,\quad \bbb,\quad \ccc,\quad \aaa+\bbb,\quad \bbb+\ccc,\quad \ccc+\aaa,\quad \aaa+\bbb+\ccc.
\end{equation}
Nontrivial classes of three qubit states obtained by (\ref{partition}) are known to be
nonempty only recently \cite{han_kye_bisep_exam}.

The main purposes of this note are twofold: Exhibiting three qubit
states in the above classes in (\ref{class}) and giving criteria for
states to be members of the cones. We will do these for so called
{\sf X}-{\sl shaped} three qubit states, whose entries are zero, by
definition, except for diagonal and anti-diagonal entries. Many
important states like GHZ diagonal states are of this form.
An {\sf X}-shaped state is also called as an \xx-{\sl state} for brevity.
In order to exhibit all the three qubit {\sf X}-states in a given cone, we
find all the extreme rays of the convex cone.
Recall that an element $x$ of a convex cone $C$ generates an extreme ray whenever
$x=x_1+x_2$ with $x_i\in C$ implies that $x_i$ is a nonnegative multiple of $x$ for $i=1,2$.
By the abuse of the terminology, we say that $x$ itself is an extreme ray in this case.
Then all the elements of convex cones in (\ref{class}) are the nonnegative sums of extreme rays.
Those extreme rays also play essential roles to find criteria for
the dual cone. Those criteria will be expressed
in terms of algebraic inequalities with the entries of {\sf X}-states, which
give rise to necessary criteria for general three qubit states to be a member of a given cone,
in terms of diagonal and anti-diagonal entries.
As for the corresponding results for full separability, we refer to the
papers \cite{{han_kye_GHZ},{han_kye_phase},{chen_han_kye},{ha-han-kye}}.

Our main tool is the duality between closed convex cones in real
vector spaces, and so, we will also consider the dual cones, whose
members play roles of witnesses, of the cones in (\ref{class}).
Since the intersection and the convex hull are dual operations, we
also naturally consider the intersections as well as convex hulls
through the discussion. Therefore, we will consider the convex cones
appearing in the following diagram
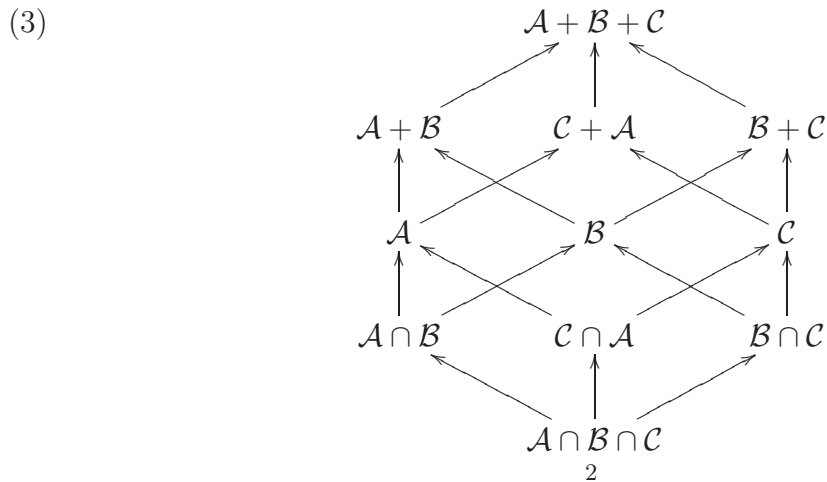
\begin{equation}\label{diagram}
\xymatrix{
& \aaa + \bbb + \ccc & \\
\aaa + \bbb \ar[ur] & \ccc + \aaa \ar[u] & \bbb + \ccc \ar[ul] \\
\aaa \ar[u] \ar[ur] & \bbb \ar[ul] \ar[ur] & \ccc \ar[ul] \ar[u] \\
\aaa \cap \bbb \ar[u] \ar[ur] & \ccc \cap \aaa \ar[ul] \ar[ur] & \bbb \cap \ccc \ar[ul] \ar[u] \\
& \aaa \cap \bbb \cap \ccc \ar[ul] \ar[u] \ar[ur] &
}
\end{equation}
which shows us partial order relations by inclusion among convex cones we are
considering. The dual cones will be also discussed.

After we explain briefly the duality in the next section, we will consider the convex cones $\aaa$,
$\bbb$ and $\ccc$ in Section \ref{sec-basic} together with their dual cones.
We will also consider the convex cones $\aaa+\bbb+\ccc$ and $\aaa\cap\bbb\cap\ccc$
in Section \ref{zec-bisep}. Conditions for those convex cones with {\sf X}-shaped matrices are already scattered in the literature
\cite{{guhne10},{gao},{Rafsanjani},{han_kye_tri},{han_kye_optimal}}.
Here, we give an alternative proof in the context of duality, together with exhibition of all
extreme rays in the convex cone of \xx-shaped matrices in each of them.
In Section \ref{sec-two}, we deal with convex hulls and intersections of two convex cones
like $\aaa+\bbb$, $\aaa\cap\bbb$ and their duals.
We will summarize our results in the final section.

The authors are grateful to the referee for careful reading and
valuable suggestions to improve presentation.

\section{duality}\label{sec-duaity}

Let $C$ be a subset of a finite dimensional real vector space $V$
with a non-degenerating bilinear pairing $\lan\ , \ \ran$, that is,
$\lan x,y\ran=0$ for every $y\in V$ implies $x=0$. We define the
{\sl dual cone} $C^\circ$ by
$$
C^\circ=\{ x\in V: \lan x,y\ran\ge 0\ {\text{\rm for every}}\ y\in C\}.
$$
Then $C^\circ$ is a closed convex cone of $V$ in general, and $C^{\circ\circ}$ is the smallest closed convex cone
containing $C$ by the Hahn-Banach type separation theorem.
If $C$ itself is a closed convex cone then we have $C=C^{\circ\circ}$, and so we see that the following are equivalent:
\begin{itemize}
\item
$x\notin C$;
\item
there exists $y\in C^\circ$ such that $\lan x,y\ran <0$.
\end{itemize}
For example, if ${\mathcal S}$ is the closed convex cone consisting of unnormalized fully separable states in the real vector space
$V$ of self-adjoint matrices in $M_2\otimes M_2 \otimes M_{2}$, then we see by this principle that
$\varrho$ is non-separable, that is, entangled if and only if there exists $W\in {\mathcal S}^\circ$ such that
$\lan W,\varrho\ran <0$. Such a $W$ must be non-positive, and called an entanglement witness \cite{terhal}.
Here, the bilinear pairing is given by $\lan a,b\ran=\tr(ba^\ttt)$ for matrices $a$ and $b$, as usual.
On the other hand, the closed convex cone ${\mathcal P}$ of all positive matrices is self-dual, that is, ${\mathcal P}^\circ={\mathcal P}$,
by the Hadamard theorem.

We note that the two operations, convex hull and intersection, are
dual to each other. In other words, the following identities
$$
(C_1+C_2)^\circ= C_1^\circ\cap C_2^\circ,\qquad
(C_1\cap C_2)^\circ = C_1^\circ + C_2^\circ
$$
hold for closed convex cones $C_1$ and $C_2$. The first identity follows from the definition. See \cite{eom-kye}.
The second one follows from the first one and the fact that the convex hull of two closed convex cones is closed.
This is an easy consequence of Carath\'eodory theorem which tells us that the convex hull of a compact set is compact.
We note that a convex cone $C$ spans the whole space $V$ if and only if $C+(-C)=V$. If we apply the above duality to the
four closed convex cones $C$, $-C$, $\{0\}$ and $V$, then we see that
the following two properties
\begin{enumerate}
\item[(${\rm C}_1$)]
$C$ spans the whole space;
\item[(${\rm C}_2$)]
$C\cap (-C)=\{0\}$
\end{enumerate}
are dual to each other. In other words, a closed convex cone $C$
satisfies (${\rm C}_1$) if and only if $C^\circ$ satisfies (${\rm C}_2$).
Recall that the real vector space $(M_{n_1}\ot M_{n_2})_{\rm sa}$
of all self-adjoint matrices in the tensor product $M_{n_1}\ot M_{n_2}$ coincides
with the tensor product $(M_{n_1})_{\rm sa}\ot (M_{n_2})_{\rm sa}$ of the self-adjoint parts.
See \cite[Section 7]{ha-han-kye}. This is also true for multi-tensor products by induction.
Therefore, the convex cone ${\mathcal S}=M_2^+\ot M_2^+\ot M_2^+$ spans the whole space
$V=(M_2\ot M_2\ot M_2)_{\rm sa}$. Since ${\mathcal P}=(M_2\ot M_2\ot M_2)^+$ satisfies
${\mathcal P}\cap (-{\mathcal P})=\{0\}$, we see that all the convex cones $\variable$
in the diagram (\ref{diagram}) also satisfy both conditions, by the relation
${\mathcal S}\subset\variable\subset{\mathcal P}$.
We list up the dual cones of the cones in (\ref{diagram}) as follows:

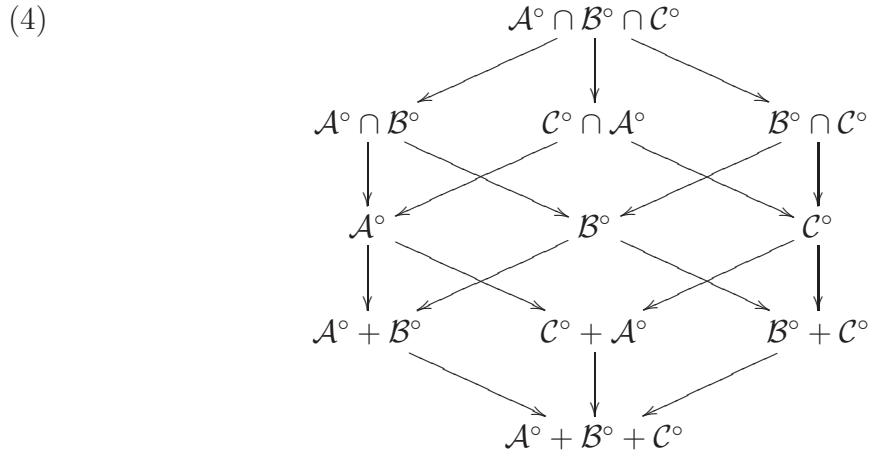
\begin{equation}\label{diagram1}
\xymatrix{
& \aaa^\circ \cap \bbb^\circ \cap \ccc^\circ \ar[dl] \ar[d] \ar[dr]& \\
\aaa^\circ \cap \bbb^\circ \ar[d] \ar[dr] & \ccc^\circ \cap \aaa^\circ \ar[dl] \ar[dr] & \bbb^\circ \cap \ccc^\circ \ar[dl] \ar[d] \\
\aaa^\circ \ar[d] \ar[dr] & \bbb^\circ \ar[dl] \ar[dr] & \ccc^\circ \ar[dl] \ar[d] \\
\aaa^\circ + \bbb^\circ \ar[dr] & \ccc^\circ + \aaa^\circ \ar[d] & \bbb^\circ + \ccc^\circ \ar[dl] \\
& \aaa^\circ + \bbb^\circ + \ccc^\circ&
}
\end{equation}

We note that all the convex cones in the diagram (\ref{diagram1}) also satisfy
both conditions (${\rm C}_1$) and (${\rm C}_2$), as dual cones of the convex cones satisfying the conditions.
An important consequence of (${\rm C}_2$) is that every
element of the convex cone is a nonnegative sum of extreme rays. See
\cite[Theorem 18.5]{rock}.

The duality is also very useful to find all the candidates for extreme rays.
We say that a subset $S$ of a closed convex cone $C$ is a {\sl generating set} for $C$ if every element of $C$ is the limit
of nonnegative sums of finitely many elements in $S$. This happens
if and only if $S^{\circ\circ}=C$ if and only if
$S^\circ=C^\circ$. In other words, we have to show that the following two statements
\begin{itemize}
\item
$y\in C^\circ$, that is, $\lan x,y\ran\ge 0$ for each $x\in C$;
\item
$\lan x,y\ran\ge 0$ for each $x\in S$
\end{itemize}
are equivalent to each other, in order to show that $S$ is a generating set for $C$.
This equivalence, in turn, gives rise to a criterion for the convex cone $C^\circ$
in terms of algebraic inequalities arising from members in the generating set $S$.
This principle will be the main tool of our discussion throughout this paper.

We note that generating sets of a convex cone are not determined uniquely.
For example, the convex cone $C$ itself is also a generating set for $C$.
Furthermore, a generating set need not contain all the extreme rays.
If a generating set $S$ for $C$ is closed, then its convex hull is also closed by
 Carath\'eodory theorem, and so every element of $C$ is the sum of finitely many elements in $S$. Therefore, we conclude that
 a {\sl closed} generating set for $C$ contains all the  extreme rays of $C$. In this way, we are looking for the set
 $\ext(C)$ of all extreme rays of the convex cone $C$.
We summarize as follows:

\begin{proposition}\label{tool}
For a subset $S$ of a closed convex cone $C$ in a finite dimensional
real vector space $V$, the following are equivalent:
\begin{enumerate}
\item[(i)]
$S$ is a generating set for $C$;
\item[(ii)]
if $y\in V$ and $\lan x,y\ran\ge 0$ for each $x\in S$, then $y\in C^\circ$.
\end{enumerate}
If $S$ is a closed generating set for $C$, then we have $\ext(C)\subset S$.
\end{proposition}

In this paper, we will concentrate on the three qubit system, and so we will work
in the real vector space $V$ of all $8\times 8$ self-adjoint matrices. The space $V$ has an important
subspace, denoted by $\xx$, consisting of all {\sf X}-shaped matrices whose entries are zero except for
diagonal and anti-diagonal entries.
In the three qubit case, an {\sf X}-shaped self-adjoint
matrix is of the form
$$
\xx(a,b,z)= \left(
\begin{matrix}
a_1 &&&&&&& z_1\\
& a_2 &&&&& z_2 & \\
&& a_3 &&& z_3 &&\\
&&& a_4&z_4 &&&\\
&&& \bar z_4& b_4&&&\\
&& \bar z_3 &&& b_3 &&\\
& \bar z_2 &&&&& b_2 &\\
\bar z_1 &&&&&&& b_1
\end{matrix}
\right),
$$
for $a,b\in\mathbb R^4$ and $z\in\mathbb C^4$, where $\mathbb C^2\ot\mathbb C^2\ot \mathbb C^2$ is identified with
the space $\mathbb C^8$ using the lexicographic order of indices.
Many important multi-qubit states arise in this form.
For example, GHZ diagonal states \cite{GHZ} are in this form, and an \xx-state $\xx(a,b,z)$
is a GHZ diagonal if and only if $a=b$ and $z\in\mathbb R^4$.

Note that $V$ and $\xx$ are of $64$ and $16$-dimensional spaces, respectively.
For a given matrix $\varrho\in V$, we denote by $\varrho_\xx$ the {\sf X}-part of $\varrho$.
The map $\varrho\mapsto \varrho_\xx$ from $V$ onto $\xx$ has the following important property.

\begin{proposition}\label{xxx-part}
For every convex cone $\variable$ in the diagram {\rm (\ref{diagram})}, we have the following:
\begin{enumerate}
\item[(i)]
if $\varrho\in\variable$, then $\varrho_{\sf X}\in\variable$;
\item[(ii)]
if $W\in\variable^\circ$, then $W_{\sf X}\in\variable^\circ$.
\end{enumerate}
\end{proposition}

\begin{proof}
It suffices to prove for the convex cone $\aaa$.
For the statement (i), it also suffices to show for a vector state $\varrho$
associated with a product vector $|x\ran\ot |y\ran\in \mathbb C^2\ot\mathbb C^4$,
where $|x\ran=(x_1,x_2)^\ttt$ and $|y\ran=(y_1,y_2,y_3,y_4)^\ttt$. We consider the following product vectors
$$
\begin{aligned}
(+x_1,+x_2)^\ttt &\ot (+y_1, +y_2, +y_3, +y_4)^\ttt,\\
(+x_1,+x_2)^\ttt &\ot (+y_1, -y_2, -y_3, +y_4)^\ttt,\\
(+x_1,-x_2)^\ttt &\ot (+y_1, -y_2, +y_3, -y_4)^\ttt,\\
(+x_1,-x_2)^\ttt &\ot (+y_1, +y_2, -y_3, -y_4)^\ttt.
\end{aligned}
$$
We take the average of four vector states associated with these four product vectors, to recover
the {\sf X}-part of $\varrho$. This proves (i). For the statement (ii) with $\variable=\aaa$, take $W\in \aaa^\circ$. For every
$\varrho\in\aaa$, we see that $\lan W_{\sf X},\varrho\ran =\lan W,\varrho_{\sf X}\ran$ is nonnegative,
because $W\in\aaa^\circ$ and $\varrho_{\sf X}\in\aaa$ by (i). Therefore, we have $W_{\sf X}\in\aaa^\circ$.
\end{proof}

Corresponding results for full separability are found in Section 3 of \cite{han_kye_GHZ}.
See also Proposition 4.1 of \cite{ha-han-kye} for multi-qubit cases.
If $\varrho$ is an {\sf X}-state, then $\lan W,\varrho\ran=\lan W_\xx,\varrho\ran$, and so we have the following:

\begin{corollary}\label{xxx-coro}
For a convex cone $\variable$ in the diagram {\rm (\ref{diagram})}, we have the following:
\begin{enumerate}
\item[(i)]
for a three qubit {\sf X}-state $\varrho$, we have $\varrho\in \variable$ if and only if $\lan W,\varrho\ran \ge 0$
for every {\sf X}-shaped $W\in \variable^\circ$;
\item[(ii)]
for a three qubit {\sf X}-shaped $W$, we have $W\in \variable^\circ$ if and only if $\lan W,\varrho\ran\ge 0$
for every {\sf X}-state $\varrho\in \variable$.
\end{enumerate}
\end{corollary}

\begin{corollary}\label{xxx24}
For convex cones $\variable_1$ and $\variable_2$ in diagrams {\rm (\ref{diagram})} or {\rm (\ref{diagram1})}, we have
the relation $(\variable_1+\variable_2)\cap\xx=(\variable_1\cap\xx)+(\variable_2\cap\xx)$.
\end{corollary}

Once we characterize {\sf X}-shaped matrices in the convex cones in (\ref{diagram}) or (\ref{diagram1}),
these conditions will give rise to necessary conditions for general three qubit self-adjoint
matrices to belong to those convex cones, by Proposition \ref{xxx-part}.
On the other hand, Corollary \ref{xxx-coro} tells us that we may restrict ourselves on the
bi-linear pairing in the real vector space $\xx$ for this purpose.


\section{basic partial separability}\label{sec-basic}

In this section, we consider the three basic convex cones $\aaa$, $\bbb$, $\ccc$ and their dual cones $\aaa^\circ$,
$\bbb^\circ$, $\ccc^\circ$.
It was shown in \cite[Proposition 5.2]{han_kye_optimal} that an {\sf X}-shaped multi-qubit state
is separable with respect to a bi-partition of systems if and only if it is of positive partial transpose
with respect to the same bi-partition.  The PPT condition is easily checked for three qubit {\sf X}-shaped states
by the following inequalities
\begin{center}
\framebox{
\parbox[t][0.7cm]{9.50cm}{
\addvspace{0.1cm} \centering
$S_1[i,j]:\qquad \min\{\sqrt{a_ib_i},\sqrt{a_jb_j}\}\ge \max\{|z_i|,|z_j|\}$,
}}
\end{center}\medskip
for $i,j=1,2,3,4$. If $\varrho=\xx(a,b,(z_1,z_2,z_3,z_4))$, then
the partial transposes are given by
$$
\begin{aligned}
\varrho^{\Gamma_A}=\xx(a,b,(\bar z_4,\bar z_3,\bar z_2,\bar z_1)),\\
\varrho^{\Gamma_B}=\xx(a,b,(z_3,z_4,z_1,z_2)),\\
\varrho^{\Gamma_C}=\xx(a,b,(z_2,z_1,z_4,z_3)).
\end{aligned}
$$
Therefore, we have the following:

\begin{proposition}\label{aa-bbb-ccc}
{\rm \cite[Proposition 5.2]{han_kye_optimal}}
For a three qubit {\sf X}-state $\varrho=\xx(a,b,z)$, we have the following:
\begin{enumerate}
\item[(i)]
$\varrho\in\aaa$ if and only if both $S_1[1,4]$ and $S_1[2,3]$ hold;
\item[(ii)]
$\varrho\in\bbb$ if and only if both $S_1[1,3]$ and $S_1[2,4]$ hold;
\item[(iii)]
$\varrho\in\ccc$ if and only if both $S_1[1,2]$ and $S_1[3,4]$ hold.
\end{enumerate}
\end{proposition}

We note that inequalities $S_1[i,j]$'s give us necessary criteria for general three qubit states to
belong to $\aaa$, $\bbb$ and $\ccc$, respectively, by Proposition \ref{xxx-part}.
Now, we proceed to provide generating sets for the convex cones $\aaa\cap\xx$, $\bbb\cap \xx$ and $\ccc\cap\xx$.
To be motivated, we decompose an {\sf X}-state $\varrho=\xx(a,b,z)$ in $\aaa$ by
$$
\begin{aligned}
\varrho &= \xx((a_1,0,0,a_4), (b_1,0,0,b_4),(z_1,0,0,z_4))\\
&\qquad\qquad\qquad\qquad + \xx((0,a_2,a_3,0), (0,b_2,b_3,0),(0,z_2,z_3,0)),
\end{aligned}
$$
then two summands satisfy both $S_1[1,4]$ and $S_1[2,3]$.
Therefore, we may assume that $a_k=b_k=z_k=0$ for $k=2,3$.
If $z_1=0$, then $\varrho$ is the average of two states
$$
\xx((a_1,0,0,a_4),(b_1,0,0,b_4),(z_4,0,0,z_4)), \quad 
\xx((a_1,0,0,a_4),(b_1,0,0,b_4),(-z_4,0,0,z_4))
$$
in $\aaa\cap \xx$. If $0<|z_1|<|z_4|$, then $\varrho$ is a convex combination of
$$
\xx((a_1,0,0,a_4),(b_1,0,0,b_4),(\textstyle{|z_4| \over |z_1|}z_1,0,0,z_4)), \quad 
\xx((a_1,0,0,a_4),(b_1,0,0,b_4),(-{|z_4| \over |z_1|}z_1,0,0,z_4))
$$
in $\aaa\cap \xx$.

By subtracting a suitable diagonal state, it is natural to consider the following conditions
\begin{center}
\framebox{
\parbox[t][0.7cm]{13.00cm}{
\addvspace{0.1cm} \centering
$S^e_1[i,j]:\qquad \sqrt{a_ib_i}=|z_i|=\sqrt{a_jb_j}=|z_j|=1$,\quad the others are zero,
}}
\end{center}\medskip
for each $i,j=1,2,3,4$ with $i\neq j$.
We define
$$
\begin{aligned}
\E_\aaa=\{\varrho=\xx(a,b,z): S^e_1[1,4] \ {\text{\rm or}}\ S^e_1[2,3]\ {\text{\rm holds}}\},\\
\E_\bbb=\{\varrho=\xx(a,b,z): S^e_1[1,3] \ {\text{\rm or}}\ S^e_1[2,4]\ {\text{\rm holds}}\},\\
\E_\ccc=\{\varrho=\xx(a,b,z): S^e_1[1,2] \ {\text{\rm or}}\ S^e_1[3,4]\ {\text{\rm holds}}\}.
\end{aligned}
$$
We also denote by $\Delta$ the set of all extreme diagonal states, that is,
$$
\Delta=\{\xx(E_i,0,0):i=1,2,3,4\}\cup\{\xx(0,E_i,0):i=1,2,3,4\},
$$
where $\{E_i:i=1,2,3,4\}$ denotes the canonical basis of $\mathbb R^4$.

We have $\E_\aaa\subset \aaa$ by Proposition \ref{aa-bbb-ccc}, and $\E_\aaa$ is parameterized by four real variables. The same comments
also hold for $\bbb$ and $\ccc$.
We also consider the following inequalities:
\begin{center}
\framebox{
\parbox[t][0.7cm]{8.00cm}{
\addvspace{0.1cm} \centering
$W_1[i,j]:\qquad \sqrt{s_it_i}+\sqrt{s_jt_j}\ge |u_i|+|u_j|$,
}}
\end{center}\medskip
for $i,j=1,2,3,4$ with $i\neq j$, in order to characterize the dual cones $\aaa^\circ\cap \xx$, $\bbb^\circ\cap \xx$ and $\ccc^\circ\cap\xx$.
We denote $\xx_i(s_i,t_i,u_i):=\xx(s_iE_i,t_iE_i,u_iE_i)$ for $i=1,2,3,4$.

\begin{lemma}\label{new-aaa}
For a given self-adjoint {\sf X}-shaped matrix  $W=\xx(s,t,u)$, the
following are equivalent:
\begin{enumerate}
\item[(i)]
$\lan W,\varrho \ran\ge 0$ for each $\varrho\in \E_\aaa\cup\Delta$;
\item[(ii)]
$s_i,t_i\ge 0$ for $i=1,2,3,4$, and the inequalities $W_1[1,4]$ and $W_1[2,3]$ hold;
\item[(iii)]
$\lan W,\varrho \ran\ge 0$ for each $\varrho\in \aaa$.
\end{enumerate}
\end{lemma}

\begin{proof}
For the direction (i) $\Longrightarrow$ (ii), we obtain $s_i,t_i\ge 0$ from
$\lan W, \varrho \ran \ge 0$ for $ \varrho \in \Delta$.
Suppose that both $s_i$ and $t_i$ are nonzero for each $i=1,2,3,4$.
Then, we can consider the following states
$$
\varrho_{i, j}:=\xx_i\left(\sqrt{t_i \over s_i}, \sqrt{s_i \over t_i}, -e^{-{\rm i}\theta_i}\right)
+ \xx_j\left(\sqrt{t_j \over s_j}, \sqrt{s_j \over t_j}, -e^{-{\rm i}\theta_j}\right)
$$
for $(i,j)=(1,4),(2,3)$, with $\theta_k = \arg u_k$.
Since $\varrho_{i,j}\in{\mathcal E}_\aaa$, we have
$$
0 \le {1 \over 2} \lan W, \varrho_{i, j} \ran = \sqrt{s_i t_i} + \sqrt{s_j t_j} - |u_i| - |u_j|,
$$
by (i). When one of $s_i$ or $t_i$ is zero, we apply the result to $W+\varepsilon I$
to get the inequality $\sqrt{(s_i+\varepsilon)(t_i+\varepsilon)} + \sqrt{(s_j+\varepsilon)(t_j+\varepsilon)}
\ge |u_i| + |u_j|$ for each $\varepsilon>0$.

For the implication (ii) $\Longrightarrow$ (iii), it is enough to prove the following by Corollary \ref{xxx-coro} and Proposition \ref{aa-bbb-ccc}:
\begin{equation}\label{miuyccj}
S_1[1,4], S_1[2,3], W_1[1,4], W_1[2,3]\ \Longrightarrow
\lan \xx(s,t,u), \xx(a,b,z)\ran\ge 0.
\end{equation}
Indeed, we have
$$
\begin{aligned}
\sum_{i=1}^4\sqrt{s_it_i}\sqrt{a_ib_i}
&=(\sqrt{s_1t_1}\sqrt{a_1b_1}+\sqrt{s_4t_4}\sqrt{a_4b_4})+(\sqrt{s_2t_2}\sqrt{a_2b_2}+\sqrt{s_3t_3}\sqrt{a_3b_3})\\
&\ge (\sqrt{s_1t_1}+\sqrt{s_4t_4})\max\{|z_1|,|z_4|\}+(\sqrt{s_2t_2}+\sqrt{s_3t_3})\max\{|z_2|,|z_3|\}\\
&\ge (|u_1|+|u_4|)\max\{|z_1|,|z_4|\}+(|u_2|+|u_3|)\max\{|z_2|,|z_3|\}\\
&\ge \sum_{i=1}^4|u_i||z_i|,
\end{aligned}
$$
which implies
\begin{equation}\label{ineq--xx}
\begin{aligned}
\frac 12\lan \xx(s,t,u),\xx(a,b,z)\ran
&=\frac 12\sum_{i=1}^4(s_ia_i+t_ib_i+2\re(u_iz_i))\\
&\ge \sum_{i=1}^4(\sqrt{s_it_i}\sqrt{a_ib_i}-|u_i||z_i|)\ge 0.
\end{aligned}
\end{equation}
as it was required. The direction (iii) $\Longrightarrow$ (i) is clear
since ${\mathcal E}_\aaa\cup\Delta\subset\aaa$ by Proposition \ref{aa-bbb-ccc}.
\end{proof}

The equivalence between (ii) and (iii) of Lemma \ref{new-aaa} gives rise to a characterization
of the convex cone $\aaa^\circ\cap \xx$, whose members are the Choi matrix of $(1,2,2)$-positive bi-linear maps between
$2\times 2$ matrices in the sense of \cite{han_kye_tri}. Therefore, Lemma \ref{new-aaa} recovers Theorem 6.2 in \cite{han_kye_tri},
as follows:

\begin{proposition}\label{basic-dual-cri}
{\rm \cite[Theorem 6.2]{han_kye_tri}}
For a self-adjoint $W=\xx(s,t,u)$ with nonnegative diagonals, we have the following:
\begin{enumerate}
\item[(i)]
$W\in\aaa^\circ$ if and only if both $W_1[1,4]$ and $W_1[2,3]$ hold;
\item[(ii)]
$W\in\bbb^\circ$ if and only if both $W_1[1,3]$ and $W_1[2,4]$ hold;
\item[(iii)]
$W\in\ccc^\circ$ if and only if both $W_1[1,2]$ and $W_1[3,4]$ hold.
\end{enumerate}
\end{proposition}

The implication (i) $\Longrightarrow$ (iii) of Lemma \ref{new-aaa} tells us that
the set $\E_\aaa\cup\Delta$ is a generating set for the convex cone $\aaa \cap \xx$ by Proposition \ref{tool}.
We also note that the set $\E_\aaa\cup\Delta$ is closed, and so we conclude
that every extreme ray of $\aaa \cap \xx$ must be an element of $\E_\aaa\cup\Delta$.
We show that the converse actually holds.
Because states in $\Delta$ generate extreme rays in the cone ${\mathcal P}$,
they also generate extreme rays of the smaller convex cones listed in the diagram (\ref{diagram}).
In order to prove that every state in the set $\E_\aaa$ generates an extreme ray of the convex cone
$\aaa \cap \xx$, we first prove a technical lemma which will play a key role in characterization of extreme rays of the other cones.

\begin{lemma}\label{CS-S}
Suppose that a three qubit \xx-state $\varrho=\xx(a,b,z)$ in $\aaa$ (respectively, $\bbb$ and $\ccc$) is decomposed as
$$
\varrho = \xx(a',b',z') + \xx(a'',b'',z'')
$$
in $\aaa \cap \xx$ (respectively, $\bbb \cap \xx$ and $\ccc \cap \xx$).
If
$$
\sqrt{a_ib_i}=\sqrt{a_jb_j}=|z_i|
$$
for $\{i,j\} = \{1,4\}$ or $\{2,3\}$ (respectively, $\{i,j\} = \{1,3\}$ or $\{2,4\}$, and $\{i,j\} = \{1,2\}$ or $\{3,4\}$), then we have
$$
(a'_i,a'_j,b'_i,b'_j,z'_i) ~ \parallelsum (a''_i,a''_j,b''_i,b''_j,z''_i).
$$
\end{lemma}

\begin{proof}
Let $k=i,j$. We have
$$
|z_i| = |z_i' + z_i''|
 \le |z_i'| + |z_i''|
 \le \sqrt{a_k'} \sqrt{b_k'} + \sqrt{a_k''} \sqrt{b_k''}
 \le \sqrt{a_k'+a_k''} \sqrt{b_k'+b_k''}
 = \sqrt{a_kb_k}
$$
by $S_1[i,j]$ and the Cauchy-Schwartz inequality.
Since $|z_i|=\sqrt{a_kb_k}$, we have
$$
z_i'~ \parallelsum z_i'', \qquad \sqrt{a_k'b_k'}=|z_i'|, \quad
\sqrt{a_k''b_k''}=|z_i''|, \quad (\sqrt{a_k'},\sqrt{a_k''}) ~
\parallelsum (\sqrt{b_k'},\sqrt{b_k''}).
$$
Let $(\sqrt{a_k'},\sqrt{a_k''}) = \lambda_k (\sqrt{b_k'},\sqrt{b_k''})$ for $\lambda_k>0$.
Then we have
$$
{z_i'' \over z_i'} = {|z_i''| \over |z_i'|} = {\sqrt{a_k''b_k''} \over \sqrt{a_k'b_k'}}
 = {\lambda_k b_k'' \over \lambda_k b_k'} = {b_k'' \over b_k'}
 = {a_k'' \slash \lambda_k^2 \over a_k' \slash \lambda_k^2} = {a_k'' \over a_k'},
$$
as it was required.
\end{proof}

\begin{theorem}\label{basic_dual}
We have
$$
{\rm Ext}(\aaa\cap\xx)=\E_\aaa\cup\Delta, \quad
{\rm Ext}(\bbb\cap\xx)=\E_\bbb\cup\Delta \quad \text{and} \quad
{\rm Ext}(\ccc\cap\xx)=\E_\ccc\cup\Delta.
$$
\end{theorem}

\begin{proof}
It suffices to show that every state in the set $\E_\aaa$ generates an extreme ray of the convex cone $\aaa \cap \xx$.
Suppose that $\varrho=\xx(a,b,z)$ satisfies the condition $S^e_1[1,4]$ and
$$
\varrho = \xx(a',b',z')+\xx(a'',b'',z'') \quad \text{in} ~\aaa\cap \xx.
$$
For $j=2,3$, we see that $a_j=b_j=0$ implies $a_j'=b_j'=z_j'=a_j''=b_j''=z_j''=0$.
Applying Lemma \ref{CS-S} with $(i,j)=(1,4)$ and $(i,j)=(4,1)$,
we get $(a',b',z') ~ \parallelsum (a'',b'',z'')$, as it was required. The same argument works for the case of $S^e_1[2,3]$.
\end{proof}

In the remainder of this section, we look for extreme rays of $\aaa^\circ \cap\xx$, $\bbb^\circ \cap\xx$ and $\ccc^\circ \cap\xx$.
To do this, we consider the condition
\begin{center}
\framebox{
\parbox[t][0.7cm]{10.00cm}{
\addvspace{0.1cm} \centering
$W^e_1[i,j]:\qquad \sqrt{s_it_i}=|u_j|=1$,\quad the others are zero,
}}
\end{center}\medskip
for $i\neq j$, and define
$$
\begin{aligned}
\E_{\aaa^\circ}&=\{W=\xx(s,t,u): W_1^e[i,j]\ {\text{\rm holds for some}}\ (i,j)=(1,4),(4,1),(2,3),(3,2)\},\\
\E_{\bbb^\circ}&=\{W=\xx(s,t,u): W_1^e[i,j]\ {\text{\rm holds for some}}\ (i,j)=(1,3),(3,1),(2,4),(4,2)\},\\
\E_{\ccc^\circ}&=\{W=\xx(s,t,u): W_1^e[i,j]\ {\text{\rm holds for some}}\ (i,j)=(1,2),(2,1),(3,4),(4,3)\}.
\end{aligned}
$$
We also consider the following set
$$
\W^\Delta=
\{W=\xx_i(r, r^{-1}, e^{{\rm i}\theta}): i=1,2,3,4,\ r>0,\ \theta\in\mathbb R\}.
$$

\begin{lemma}\label{new-aaa-circ}
For a given self-adjoint {\sf X}-shaped matrix $\varrho=\xx(a,b,z)$, the
following are equivalent:
\begin{enumerate}
\item[(i)]
$\lan W,\varrho \ran\ge 0$ for each $W\in \E_{\aaa^\circ}\cup\Delta\cup\W^\Delta$;
\item[(ii)]
$\varrho$ is a state satisfying the inequalities $S_1[1,4]$ and $S_1[2,3]$;
\item[(iii)]
$\lan W,\varrho \ran\ge 0$ for each $W\in \aaa^\circ$.
\end{enumerate}
\end{lemma}

\begin{proof}
The equivalence between (ii) and (iii) follows from Proposition \ref{aa-bbb-ccc}.
Therefore, it suffices to show the direction (i) $\Longrightarrow$ (ii).
Since $\lan W,\varrho \ran\ge 0$ for $W \in \Delta$, we have $a_i,b_i \ge 0$.
By taking $\varrho+\varepsilon I$ into account as in the proof of Lemma \ref{new-aaa},
we may assume that $a_i,b_i >0$ without loss of generality.
Then, we can consider
$$
W_{i,j}:=\xx_i\left(\sqrt{\frac{b_i}{a_i}},\sqrt{\frac{a_i}{b_i}},0\right)+\xx_j\left(0,0,-e^{{\rm i}\theta_j}\right)
\in \E_{\aaa^\circ}\cup\W^\Delta
$$
for $(i,j)=(1,4),(4,1),(2,3),(3,2)$ or $i=j$. Note that $W_{i,i}\in{\mathcal W}^\Delta$.
We see that $\varrho$ is a state by $\lan W_{i,i}, \varrho \ran \ge 0$, and the inequalities $S_1[1,4]$ and $S_1[2,3]$
follow from $\lan W_{i,j}, \varrho \ran \ge 0$ for $(i,j)=(1,4),(4,1),(2,3),(3,2)$.
\end{proof}

As for extreme rays of the dual cones, we also begin with a technical lemma which is a witness counterpart to Lemma \ref{CS-S}.

\begin{lemma}\label{CS-W}
Suppose that a three qubit self-adjoint \xx-shaped matrix $W=\xx(s,t,u)$ in $\aaa^\circ$
(respectively, $\bbb^\circ$ and $\ccc^\circ$) is decomposed as
$$
W = \xx(s',t',u') + \xx(s'',t'',u'')
$$
in $\aaa^\circ \cap \xx$ (respectively, $\bbb^\circ \cap \xx$ and $\ccc^\circ \cap \xx$).
If
$$
\sqrt{s_it_i}=|u_k| \qquad \text{and} \qquad  s_j=t_j=0
$$
for $\{i,j\} = \{1,4\}$ or $\{2,3\}$ (respectively, $\{i,j\} = \{1,3\}$ or $\{2,4\}$,
and $\{i,j\} = \{1,2\}$ or $\{3,4\}$) and $k \in \{i,j\}$, then we have
$$
s_j'=t_j'=0=s_j''=t_j'', \qquad u_\ell'=0=u_\ell'' \qquad \text{and} \qquad (s'_i,t'_i,u'_k) ~ \parallelsum (s''_i,t''_i,u''_k)
$$
for $\ell$ with $\{k,\ell\}=\{i,j\}$.
\end{lemma}

\begin{proof}
The condition $s_j=t_j=0$  implies $s_j'=t_j'=0=s_j''=t_j''$.
We have
$$
\begin{aligned}
|u_k|
&= |u_k' + u_k''|\\
&\le |u_k'| + |u_k''|\\
&\le \sqrt{s_i'} \sqrt{t_i'} - |u_\ell'| + \sqrt{s_i''} \sqrt{t_i''} - |u_\ell''|\\
&\le \sqrt{s_i'+s_i''} \sqrt{t_i'+t_i''} - |u_\ell'| - |u_\ell''| \\
& = \sqrt{s_it_i} - |u_\ell'| - |u_\ell''| \\
& \le \sqrt{s_it_i}
\end{aligned}
$$
by $W_1[i,j]=W_1[k,\ell]$ and the Cauchy-Schwartz inequality.
Since $|u_k|=\sqrt{s_it_i}$, we have
$$
u_k'~ \parallelsum ~u_k'', \quad u_\ell'=0=u_\ell'' \quad \sqrt{s_i't_i'}=|u_k'|, \quad \sqrt{s_i''t_i''}=|u_k''|, \quad
(\sqrt{s_i'},\sqrt{s_i''}) ~ \parallelsum ~(\sqrt{t_i'},\sqrt{t_i''}).
$$
Let $(\sqrt{s_i'},\sqrt{s_i''}) = \lambda (\sqrt{t_i'},\sqrt{t_i''})$ for $\lambda>0$.
Then we have
$$
{u_k'' \over u_k'} = {|u_k''| \over |u_k'|} = {\sqrt{s_i''t_i''} \over \sqrt{s_i't_i'}}
= {\lambda t_i'' \over \lambda t_i'} = {t_i'' \over t_i'} = {s_i'' \slash \lambda^2 \over s_i' \slash \lambda^2} = {s_i'' \over s_i'}.
$$
\end{proof}

\begin{theorem}\label{basic-dual-ext}
We have
$$
{\rm Ext}(\aaa^\circ\cap\xx)=\E_{\aaa^\circ}\cup\Delta\cup\W^\Delta, \quad
{\rm Ext}(\bbb^\circ\cap\xx)=\E_{\bbb^\circ}\cup\Delta\cup\W^\Delta, \quad
{\rm Ext}(\ccc^\circ\cap\xx)=\E_{\ccc^\circ}\cup\Delta\cup\W^\Delta.
$$
\end{theorem}

\begin{proof}
It suffices to show that every ray in $\E_{\aaa^\circ}\cup\Delta\cup\W^\Delta$ is extreme by Proposition \ref{tool}.
It is easy to see that diagonal states in $\Delta$ generate extreme rays of the convex cone $\aaa^\circ\cap\xx$ by the conditions
$W_1[1,4]$ and $W_1[2,3]$.
For the remaining cases for $\E_{\aaa^\circ}\cup\W^\Delta$, we take $i=1,4$ and may assume that $W=\xx(s,t,u)$ satisfies
$$
\sqrt{s_1t_1}=|u_i|>0 \qquad \text{and} \qquad s_k=t_k=u_\ell=0\ {\text{\rm for}}\ k \ne 1\ {\text{\rm and}}\ \ell \ne i.
$$
Suppose that
$W = \xx(s',t',u')+\xx(s'',t'',u'')$ in $\aaa^\circ \cap \xx$.
For $k=2,3,4$, the condition $s_k=t_k=0$  implies
$s_k'=t_k'=0=s_k''=t_k''$.
Combining this with $W_1[2,3]$, we also have
$u_2'=u_3'=0=u_2''=u_3''$.
Applying Lemma \ref{CS-W} with $(i,j,k,\ell)=(1,4,1,4)$, $(1,4,4,1)$,
we get $(s',t',u') ~ \parallelsum (s'',t'',u'')$.
\end{proof}

\section{Full bi-separability and bi-separability}\label{zec-bisep}

In this section, we consider convex cones
$\aaa\cap\bbb\cap\ccc$ for full bi-separable states and $\aaa+\bbb+\ccc$ for bi-separable states,
together with their dual cones
$\aaa^\circ+\bbb^\circ+\ccc^\circ$ and $\aaa^\circ\cap\bbb^\circ\cap\ccc^\circ$, respectively.
We first note  that $\varrho=\xx(a,b,z)\in\aaa\cap\bbb\cap\ccc$
if and only if $S_1[i,j]$ holds for every $i,j=1,2,3,4$,
which is equivalent to the PPT condition of $\varrho$ \cite[Theorem 5.3]{han_kye_optimal}.
In order to find extreme rays of the cone $\aaa\cap\bbb\cap\ccc\cap\xx$,
we consider the condition
\begin{center}
\framebox{
\parbox[t][0.7cm]{8.00cm}{
\addvspace{0.1cm} \centering
$S^e_3:\qquad \sqrt{a_ib_i}=|z_i|=1, \ i=1,2,3,4,$
}}
\end{center}\medskip
and define
$$
\E_{\aaa\cap\bbb\cap\ccc}=\{\varrho=\xx(a,b,z):S^e_3\ {\text{\rm holds}}\}.
$$
We also recall the inequality
\begin{center}
\framebox{
\parbox[t][0.7cm]{7.00cm}{
\addvspace{0.1cm} \centering
$W_3 :\qquad \sum_{i=1}^4\sqrt{s_it_i}\ge\sum_{i=1}^4|u_i|,$
}}
\end{center}\medskip
which appears in the characterization of decomposability of \xx-shaped entanglement witnesses in \cite[Theorem 5.5]{han_kye_optimal}.

\begin{lemma}\label{lemma-ppt-dual}
For a given self-adjoint {\sf X}-shaped matrix $W=\xx(s,t,u)$, the following are equivalent.
\begin{enumerate}
\item[(i)]
$\lan W,\varrho\ran\ge 0$ for each $\varrho\in \E_{\aaa\cap\bbb\cap\ccc}\cup\Delta$;
\item[(ii)]
$s_i,t_i\ge 0$ for $i=1,2,3,4$, and the inequality $W_3$ holds;
\item[(iii)]
$\lan W,\varrho\ran\ge 0$ for each $\varrho\in\aaa\cap\bbb\cap\ccc$.
\end{enumerate}
\end{lemma}

\begin{proof}
For the direction (i) $\Longrightarrow$ (ii), we first
obtain $s_i,t_i\ge 0$ from $\lan W, \varrho \ran \ge 0$ for $ \varrho \in \Delta$.
In order to prove the inequality $W_3$, we may assume that $s_i,t_i>0$ as in the proof of Lemma \ref{new-aaa}.
We can consider the state $\varrho$ defined by
$$
\varrho:=\xx\left(\left(r_1,r_2,r_3,r_4\right),
   \left(r_1^{-1}, r_2^{-1}, r_3^{-1}, r_4^{-1}\right),
   \left(-e^{-{\rm i}\theta_1},-e^{-{\rm i}\theta_2},-e^{-{\rm i}\theta_3},-e^{-{\rm i}\theta_4}\right)\right)
$$
with $r_i=\sqrt{t_i\over s_i}$ and $\theta_k = \arg u_k$.
This state $\varrho$ belongs to $\E_{\aaa\cap\bbb\cap\ccc}$, and so $\lan W,\varrho\ran\ge 0$ gives rise to
the inequality $W_3$.
For (ii) $\Longrightarrow$ (iii), it suffices to
show that $\lan W,\varrho\ran\ge 0$ for $\varrho=\xx(a,b,z)$ satisfying $S_1[i,j]$ for all $i,j$ and $W=\xx(s,t,u)$
satisfying $W_3$ by Corollary \ref{xxx-coro}. Indeed, taking $M$ satisfying
$\sqrt{a_ib_i}\ge M\ge |z_j|$ for each $i,j=1,2,3,4$, we have
$$
\sum_{i=1}^4\sqrt{s_it_i}\sqrt{a_ib_i}\ge M\sum_{i=1}^4\sqrt{s_it_i}\ge M\sum_{i=1}^4|u_i|\ge\sum_{i=1}^4|z_i||u_i|,
$$
which implies $\lan W,\varrho\ran\ge 0$, as in (\ref{ineq--xx}).
\end{proof}

Since $(\aaa\cap\bbb\cap\ccc)^\circ=\aaa^\circ+\bbb^\circ+\ccc^\circ$,
the equivalence (ii) $\Longleftrightarrow$ (iii) in Lemma \ref{lemma-ppt-dual}
gives another proof for \cite[Theorem 5.5]{han_kye_optimal} which uses the duality principle.

\begin{proposition}\label{coro_decomp}\cite[Theorem 5.5]{han_kye_optimal}
An {\sf X}-shaped self-adjoint matrix $W=\xx(s,t,u)$ with nonnegative diagonals belongs to $\aaa^\circ+\bbb^\circ+\ccc^\circ$ if and only if
the inequality $W_3$ holds.
\end{proposition}

For convex cones $C_1$ and $C_2$, it is clear that $\ext(C_1+C_2)\subset \ext(C_1)\cup\ext(C_2)$
in general. Therefore, we see that
$\ext((\aaa^\circ+\bbb^\circ+\ccc^\circ)\cap\xx)$ is contained in the union of
$\ext(\aaa^\circ\cap\xx)$, $\ext(\bbb^\circ\cap\xx)$ and $\ext(\ccc^\circ\cap\xx)$ by Corollary \ref{xxx24}.
We show that they actually coincide.

\begin{theorem}\label{ppt-ext}
We have the following:
\begin{enumerate}
\item[(i)]
$\ext(\aaa\cap\bbb\cap\ccc\cap\xx)=\E_{\aaa\cap\bbb\cap\ccc}\cup\Delta$;
\item[(ii)]
$\ext((\aaa^\circ+\bbb^\circ+\ccc^\circ)\cap\xx)=
\ext(\aaa^\circ\cap\xx)\cup
\ext(\bbb^\circ\cap\xx)\cup
\ext(\ccc^\circ\cap\xx)$.
\end{enumerate}
\end{theorem}

\begin{proof}
For (i), it remains to show that
every PPT state in $\E_{\aaa\cap\bbb\cap\ccc}$ generates an extreme ray of the cone
$\aaa\cap\bbb\cap\ccc$.
Suppose that $\varrho=\xx(a,b,z)$ satisfies the condition $S^e_3$ and
$$
\varrho = \xx(a',b',z')+\xx(a'',b'',z'') \quad \text{in} ~\aaa\cap\bbb\cap\ccc\cap \xx.
$$
Applying Lemma \ref{CS-S} with all the pairs $(i,j)$,
we conclude
$(a',b',z') ~ \parallelsum (a'',b'',z'')$,
and this completes the proof of (i).

In order to prove (ii), it suffices to show
$\ext(\aaa^\circ\cap\xx)\subset\ext((\aaa^\circ+\bbb^\circ+\ccc^\circ)\cap\xx)$.
It is easy to see that diagonal states in $\Delta$ generate extreme rays
in the convex cone $(\aaa^\circ+\bbb^\circ+\ccc^\circ)\cap\xx$ by the condition $W_3$.
We will show that $W=\xx(rE_1,r^{-1}E_1,e^{{\rm i}\theta}E_j)$ generates an extreme ray of the cone
$(\aaa^\circ+\bbb^\circ+\ccc^\circ)\cap\xx$ for $j=1,4$. Suppose that
\begin{equation}\label{jhjgjgjhj}
W=\xx(s^\prime,t^\prime,u^\prime)+\xx(s^{\pr\pr},t^{\pr\pr},u^{\pr\pr})\qquad {\text{\rm in}}\ (\aaa^\circ+\bbb^\circ+\ccc^\circ)\cap\xx.
\end{equation}
For $i=2,3,4$, $s_i=0=t_i$ implies that $s_i'=t_i'=0=s_i''=t_i''$.
By $W_3$, we have
$$
\begin{aligned}
1=|e^{{\rm i}\theta}|=|u_j^\pr+u_j^{\pr\pr}|\le
|u^\pr_j|+|u^{\pr\pr}_j|
&\le\sum_{k=1}^4(|u^\pr_k|+|u^{\pr\pr}_k|)\\
&\le\sqrt{s^\pr_1t^\pr_1}+\sqrt{s^{\pr\pr}_1t^{\pr\pr}_1}
\le\sqrt{s_1^\pr+s_1^{\pr\pr}}\sqrt{t_1^\pr+t_1^{\pr\pr}}=1,
\end{aligned}
$$
and so it follows that $u^\pr_k=u^{\pr\pr}_k=0$ for $k\neq j$. Therefore, the summands in (\ref{jhjgjgjhj})
belong to the cone $\aaa^\circ\cap \xx$ by $W_3$ again, and we may apply Theorem \ref{basic-dual-ext}.
\end{proof}

Now, we turn our attention to the cone $\aaa^\circ\cap\bbb^\circ\cap\ccc^\circ$ and its dual cone.
For each $i=1,2,3,4$, we consider the condition
\begin{center}
\framebox{
\parbox[t][0.7cm]{13.00cm}{
\addvspace{0.1cm} \centering
$W^e_3[i]:\qquad |u_i|=\sqrt{s_jt_j}=\sqrt{s_kt_k}=\sqrt{s_\ell t_\ell}=1$,\quad the others are zero,
}}
\end{center}\medskip
where $j,k,\ell$ are chosen so that $i,j,k,\ell$ are mutually distinct, and define
$$
\E_{\aaa^\circ\cap\bbb^\circ\cap\ccc^\circ}
=\{W=\xx(s,t,u): W^e_3[i]\ {\text{\rm holds for some}}\ i=1,2,3,4\}.
$$
We also consider the following inequality
\begin{center}
\framebox{
\parbox[t][1.1cm]{8.00cm}{
\addvspace{0.1cm} \centering
$S_3 :\qquad \displaystyle{\sum_{j\neq i}\sqrt{a_jb_j}\ge |z_i|},\quad i=1,2,3,4$.
}}
\end{center}\medskip
These are exactly the inequalities which appear in the necessary criteria \cite{guhne10} for bi-separability.
We also refer to \cite{gao} for necessary criteria of multi-qubit bi-separable states.
If $\varrho$ itself is {\sf X}-shaped, then the converse is also true \cite{Rafsanjani}.
The authors have shown in \cite[Corollary 3.4]{han_kye_optimal} that even a PPT mixture
satisfies the multi-qubit analogue of $S_3$, to recover the above characterization of
bi-separability of multi-qubit {\sf X}-states.
We give here another alternative proof using the duality.

\begin{lemma}\label{lemma-ppt-dual---}
For a given self-adjoint {\sf X}-shaped matrix $\varrho=\xx(a,b,z)$, the following are equivalent.
\begin{enumerate}
\item[(i)]
$\lan W,\varrho\ran\ge 0$ for each $W\in \E_{\aaa^\circ\cap\bbb^\circ\cap\ccc^\circ}\cup\W^\Delta\cup\Delta$;
\item[(ii)]
$\varrho$ is a state satisfying the inequality $S_3$;
\item[(iii)]
$\lan W,\varrho\ran\ge 0$ for each $W\in \aaa^\circ\cap\bbb^\circ\cap\ccc^\circ$.
\end{enumerate}
\end{lemma}

\begin{proof}
For the direction (i) $\Longrightarrow$ (ii), we first note that
$\varrho$ is a state as in the proof of Lemma \ref{new-aaa-circ}. Now, we consider
$$
W:=\xx\left(\left(0,\sqrt{\frac{a_2}{b_2}},\sqrt{\frac{a_3}{b_3}},\sqrt{\frac{a_4}{b_4}},\right),
\left(0,\sqrt{\frac{b_2}{a_2}},\sqrt{\frac{b_3}{a_3}},\sqrt{\frac{b_4}{a_4}},\right),
\left(-e^{-{\rm i}\theta_1},0,0,0\right)\right),
$$
which belongs to $\E_{\aaa^\circ\cap\bbb^\circ\cap\ccc^\circ}$,
where $\theta_1=\arg z_1$. Then, we have
$$
0 \le {1 \over 2} \lan \varrho, W \ran = \sum_{j\neq 1} \sqrt{a_j b_j} - |z_1|.
$$
The other inequalities come out by the same way.

For the direction (ii) $\Longrightarrow$ (iii),
it suffices to show the following:
$$
S_3,\ W_1[i,j]\ {\text{\rm for}}\ i,j=1,2,3,4\ {\text{\rm with}}\ i\neq j \
\Longrightarrow\ \lan \xx(s,t,u),\xx(a,b,z)\ran\ge 0
$$
by Corollary \ref{xxx-coro} and Proposition \ref{basic-dual-cri}.
The inequality $\lan W,\varrho\ran\ge 0$ is trivial when $W$ is positive, that is, $\sqrt{s_i t_i} \ge |u_i|$ for all $i=1,2,3,4$.
Suppose that $W$ is not positive, and so there exists $i_0\in\{1,2,3,4\}$ such that
$\sqrt{s_{i_0} t_{i_0}} < |u_{i_0}|$, say $\sqrt{s_{1} t_{1}} < |u_{1}|$ without loss of generality.
We have
$$
(\sqrt{s_1t_1}+\sqrt{s_it_i})\sqrt{a_ib_i}\ge (|u_1|+|u_i|)\sqrt{a_ib_i},\qquad i=2,3,4,
$$
by $W_1[1,i]$. Summing up, we also have
$$
\begin{aligned}
\sqrt{s_1t_1}\sum_{i=2}^4\sqrt{a_ib_i}+\sum_{i=2}^4\sqrt{s_it_i}\sqrt{a_ib_i}
&\ge |u_1|\sum_{i=2}^4\sqrt{a_ib_i}+\sum_{i=2}^4|u_i|\sqrt{a_ib_i}\\
&\ge |u_1|\sum_{i=2}^4\sqrt{a_ib_i}+\sum_{i=2}^4|u_i||z_i|,
\end{aligned}
$$
which implies
$$
\begin{aligned}
\sum_{i=2}^4\left(\sqrt{s_it_i}\sqrt{a_ib_i}-|u_i||z_i|\right)
&\ge \left(|u_1|-\sqrt{s_1t_1}\right)\sum_{i=2}^4\sqrt{a_ib_i}\\
&\ge \left(|u_1|-\sqrt{s_1t_1}\right)|z_1|
\ge |u_1||z_1|-\sqrt{s_1t_1}\sqrt{a_1b_1},
\end{aligned}
$$
by $S_3$ and $\sqrt{s_{1} t_{1}} < |u_{1}|$.
Therefore, we have $\sum_{i=1}^4\sqrt{s_it_i}\sqrt{a_ib_i}\ge \sum_{i=1}^4|u_i||z_i|$,
which completes the proof by (\ref{ineq--xx}).
\end{proof}

Since the dual cone of $\aaa^\circ\cap\bbb^\circ\cap\ccc^\circ$ is just $\aaa+\bbb+\ccc$, we recover
the following characterization of biseparability of three qubit states.
Especially, every three qubit biseparable state with the {\sf X}-part $\xx(a,b,z)$ must satisfy
the inequalities $S_3$, as it was observed in \cite{guhne10}.

\begin{proposition}\label{i_join}\cite{{guhne10},{han_kye_optimal},{Rafsanjani}}
For a three qubit \xx-state $\varrho=\xx(a,b,z)$, the following are equivalent:
\begin{enumerate}
\item[(i)]
$\varrho$ belongs to $\aaa+\bbb+\ccc$;
\item[(ii)]
the inequality $S_3$ holds.
\end{enumerate}
\end{proposition}

As for extreme rays, we also have the following:

\begin{theorem}\label{bi-sep-ext}
We have the following:
\begin{enumerate}
\item[(i)]
$\ext((\aaa+\bbb+\ccc)\cap\xx)=\ext(\aaa\cap\xx)\cup\ext(\bbb\cap\xx)\cup\ext(\ccc\cap\xx)$;
\item[(ii)]
$\ext(\aaa^\circ\cap\bbb^\circ\cap\ccc^\circ\cap\xx)=\E_{\aaa^\circ\cap\bbb^\circ\cap\ccc^\circ}\cup\W^\Delta\cup\Delta$.
\end{enumerate}
\end{theorem}

\begin{proof}
For (i), it suffices to show
$\ext(\aaa\cap\xx)\subset\ext((\aaa+\bbb+\ccc)\cap\xx)$.
Suppose that $\varrho=\xx(a,b,z)$ satisfies the condition $S^e_1[1,4]$, and
\begin{equation}\label{,lonkhvgfsd}
\varrho = \xx(a',b',z')+\xx(a'',b'',z'') \quad \text{in} ~(\aaa+\bbb+\ccc)\cap \xx.
\end{equation}
Then we have $a_i^\pr=a_i^{\pr\pr}=b_i^\pr=b_i^{\pr\pr}=0$ for $i=2,3$, which also implies that
$z_i^\pr=z_i^{\pr\pr}=0$ for $i=2,3$. By the inequality $S_3$, the summands in (\ref{,lonkhvgfsd}) must belong
to the cone $\aaa$. Therefore, we can apply Theorem \ref{basic_dual}.

As for (ii), we note that matrices in $\Delta$ and $\W^\Delta$ generate extreme rays in $\aaa^\circ \cap \xx$,
and so they also generate extreme rays in the smaller cone $\aaa^\circ \cap \bbb^\circ \cap \ccc^\circ \cap \xx$.
Suppose that $W=\xx(s,t,u)$ satisfies $W_3^e[1]$ and
$$
W = \xx(s',t',u')+\xx(s'',t'',u'') \quad \text{in} ~\aaa^\circ \cap \bbb^\circ \cap \ccc^\circ \cap \xx.
$$
Applying Lemma \ref{CS-W} with $(i,j,k,\ell)=(2,1,1,2)$,$(3,1,1,3)$,$(4,1,1,4)$, we get $(s',t',u') ~ \parallelsum (s'',t'',u'')$.
\end{proof}

It was shown in \cite[Theorem 4.1]{han_kye_optimal} that $W\in\E_{\aaa^\circ\cap\bbb^\circ\cap\ccc^\circ}$
is an optimal genuine entanglement witness. This means that the set
$\{\varrho\in{\mathcal P}:\lan W,\varrho\ran <0\}$ of genuine entanglement detected by $W$
is maximal with respect to the inclusion. It is easy to see that extremeness implies optimality.
We have shown in Theorem \ref{bi-sep-ext} that $W\in\E_{\aaa^\circ\cap\bbb^\circ\cap\ccc^\circ}$ is extreme in the cone
$\aaa^\circ\cap\bbb^\circ\cap\ccc^\circ\cap\xx$. It would be interesting to ask if they are extreme
in the much bigger convex cone $\aaa^\circ\cap\bbb^\circ\cap\ccc^\circ$.

\section{intersections and convex hulls of two basic cones}\label{sec-two}

In this section, we consider the following convex cones
$$
\aaa\cap\bbb,\quad \bbb\cap\ccc,\quad \ccc\cap\aaa,\quad \aaa+\bbb,\quad \bbb+\ccc, \quad \ccc+\aaa,
$$
together with their dual cones:
$$
\aaa^\circ+\bbb^\circ,\quad \bbb^\circ+\ccc^\circ, \quad \ccc^\circ+\aaa^\circ,\quad
\aaa^\circ\cap\bbb^\circ,\quad \bbb^\circ\cap\ccc^\circ,\quad \ccc^\circ\cap\aaa^\circ.
$$
We look for inequalities characterizing the above convex cones, together with extreme rays
of the cones. As for intersections of two cones, we just put
together inequalities for both cones. For a three qubit {\sf X}-state $\varrho=\xx(a,b,z)$, we have the following:
\begin{itemize}
\item
$\varrho\in\aaa\cap\bbb$ if and only if $S_1[1,4], S_1[2,3], S_1[1,3], S_1[2,4]$ hold;
\item
$\varrho\in\bbb\cap\ccc$ if and only if $S_1[1,3], S_1[2,4], S_1[1,2], S_1[3,4]$ hold;
\item
$\varrho\in\ccc\cap\aaa$ if and only if $S_1[1,2], S_1[3,4], S_1[1,4], S_1[2,3]$ hold.
\end{itemize}
For an {\sf X}-shaped $W=\xx(s,t,u)$, we also have
\begin{itemize}
\item
$W\in\aaa^\circ\cap\bbb^\circ$ if and only if $W_1[1,4], W_1[2,3], W_1[1,3], W_1[2,4]$ hold;
\item
$W\in\bbb^\circ\cap\ccc^\circ$ if and only if $W_1[1,3], W_1[2,4], W_1[1,2], W_1[3,4]$ hold;
\item
$W\in\ccc^\circ\cap\aaa^\circ$ if and only if $W_1[1,2], W_1[3,4], W_1[1,4], W_1[2,3]$ hold.
\end{itemize}

In order to find extreme rays of the cones $\bbb\cap\ccc\cap\xx$, $\ccc\cap\aaa\cap\xx$ and $\aaa\cap\bbb\cap\xx$,
we consider the condition
\begin{center}
\framebox{
\parbox[t][0.7cm]{14.00cm}{
\addvspace{0.1cm} \centering
$S^e_2[i,j]:\qquad |z_i|=\sqrt{a_ib_i}=\sqrt{a_kb_k}=\sqrt{a_\ell b_\ell}=1$,\quad the others are zero,
}}
\end{center}\medskip
for $i,j=1,2,3,4$ with $i\neq j$, where $k,\ell$ are chosen so that $i,j,k,\ell$ are mutually distinct.
Here, we point out that $a_j=0=b_j$.
We define
$$
\begin{aligned}
\E_{\bbb\cap\ccc}&=\{\varrho=\xx(a,b,z): S^e_2[i,j]\ {\text{\rm holds for some}}\
    (i,j)=(1,4), (4,1),(2,3), (3,2)\},\\
\E_{\ccc\cap\aaa}&=\{\varrho=\xx(a,b,z): S^e_2[i,j]\ {\text{\rm holds for some}}\
    (i,j)=(1,3), (3,1),(2,4),(4,2)\},\\
\E_{\aaa\cap\bbb}&=\{\varrho=\xx(a,b,z): S^e_2[i,j]\ {\text{\rm holds for some}}\
    (i,j)=(1,2), (2,1),(3,4), (4,3)\},
\end{aligned}
$$
and consider the following inequalities
\begin{center}
\framebox{
\parbox[t][0.7cm]{10.50cm}{
\addvspace{0.1cm} \centering
${W_2[i,j]}:\qquad \sum_{k\neq j}\sqrt{s_kt_k}\ge |u_i|,\quad \sum_{k\neq i}\sqrt{s_kt_k}\ge |u_j|$,
}}
\end{center}\medskip
for $i,j=1,2,3,4$ with $i\neq j$.

\begin{lemma}\label{lemma-ntersection-state}
For a given self-adjoint {\sf X}-shaped matrix $W=\xx(s,t,u)$, the following are equivalent.
\begin{enumerate}
\item[(i)]
$\lan W,\varrho\ran\ge 0$ for each $\varrho\in \E_{\bbb\cap\ccc}\cup \E_{\aaa\cap\bbb\cap\ccc}\cup\Delta$
{\rm (}respectively, $\E_{\ccc\cap\aaa}\cup \E_{\aaa\cap\bbb\cap\ccc}\cup\Delta$ and
$\E_{\aaa\cap\bbb}\cup \E_{\aaa\cap\bbb\cap\ccc}\cup\Delta${\rm )};
\item[(ii)]
$s_i,t_i\ge 0$ for $i=1,2,3,4$ and the inequalities $W_3$ and $W_2[1,4]$, $W_2[2,3]$
{\rm (}respectively, $W_2[1,3],W_2[2,4]$ and $W_2[1,2],W_2[3,4]${\rm )} hold;
\item[(iii)]
$\lan W,\varrho\ran\ge 0$ for each $\varrho\in\bbb\cap\ccc$
{\rm (}respectively, $\ccc\cap\aaa$ and $\aaa\cap\bbb${\rm )}.
\end{enumerate}
\end{lemma}

\begin{proof}
The inequalities $s_i,t_i\ge 0$ and $W_3$ follow from Lemma \ref{lemma-ppt-dual}.
We will prove for $\bbb\cap\ccc$. The others follow by applying the operator
$x_A\ot x_B\ot x_C\mapsto x_{\sigma(A)}\ot x_{\sigma(B)}\ot x_{\sigma(C)}$ for permutations $\sigma$ on $\{A,B,C\}$.
To prove (i) $\Longrightarrow$ (ii), we may assume that all the diagonal elements $s_i$ and $t_i$ are nonzero, and
consider four \xx-states
$$
\begin{aligned}
\varrho_{1,4}:&=\xx\left(\left(\sqrt{t_1 \over s_1},\sqrt{t_2 \over s_2},\sqrt{t_3 \over s_3},0\right),
   \left(\sqrt{s_1 \over t_1},\sqrt{s_2 \over t_2},\sqrt{s_3 \over t_3},0\right),\left(-e^{-{\rm i}\theta_1},0,0,0\right)\right), \\
\varrho_{2,3}:&=\xx\left(\left(\sqrt{t_1 \over s_1},\sqrt{t_2 \over s_2},0,\sqrt{t_4 \over s_4}\right),
   \left(\sqrt{s_1 \over t_1},\sqrt{s_2 \over t_2},0,\sqrt{s_4 \over t_4}\right),\left(0,-e^{-{\rm i}\theta_2},0,0\right)\right), \\
\varrho_{3,2}:&=\xx\left(\left(\sqrt{t_1 \over s_1},0,\sqrt{t_3 \over s_3},\sqrt{t_4 \over s_4}\right),
   \left(\sqrt{s_1 \over t_1},0,\sqrt{s_3 \over t_3},\sqrt{s_4 \over t_4}\right),\left(0,0,-e^{-{\rm i}\theta_3},0\right)\right),\\
\varrho_{4,1}:&=\xx\left(\left(0,\sqrt{t_2 \over s_2},\sqrt{t_3 \over s_3},\sqrt{t_4 \over s_4}\right),
   \left(0,\sqrt{s_2 \over t_2},\sqrt{s_3 \over t_3},\sqrt{s_4 \over t_4}\right),\left(0,0,0,-e^{-{\rm i}\theta_4}\right)\right),
\end{aligned}
$$
with $\theta_k = \arg u_k$.
These states belong to $\E_{\bbb\cap\ccc}$.
We expand $\langle W, \varrho_{i,j} \rangle \ge 0$ to obtain  $W_2[1,4]$ and $W_2[2,3]$.

For (ii) $\Longrightarrow$ (iii), it suffices to show $\lan W,\varrho\ran\ge 0$
when $W=\xx(s,t,u)$ satisfies $W_2[1,4]$, $W_2[2,3]$, $W_3$, and $\varrho=\xx(a,b,z)$ satisfies
$S_1[1,3], S_1[2,4], S_1[1,2], S_1[3,4]$ by Corollary \ref{xxx-coro} and Proposition \ref{aa-bbb-ccc}.
If $\varrho\in\aaa\cap\bbb\cap\ccc$, then this is trivial by $W_3$ and Proposition \ref{coro_decomp}.
So, we may assume that $\varrho\notin\aaa$,
especially $|z_4| > \sqrt{a_1b_1}$, without loss of generality.
We begin with
$$
\begin{aligned}
{1 \over 2}\lan W, \varrho \ran
&\ge\sum_{i=1}^4 \left(\sqrt{s_i t_i} \sqrt{a_i b_i} - |u_i| |z_i|\right)\\
& = \left( \sqrt{s_2t_2}\sqrt{a_2b_2} + \sqrt{s_3t_3}\sqrt{a_3b_3} + \sqrt{s_4t_4}\sqrt{a_4b_4} -|u_4||z_4|\right) \\
& \qquad\qquad +\left( \sqrt{s_1t_1}\sqrt{a_1b_1} -|u_1||z_1| -|u_2||z_2| -|u_3||z_3|\right), \\
\end{aligned}
$$
as in (\ref{ineq--xx}).
We have $\sqrt{a_ib_i}\ge |z_4|$ for $i=2,3,4$ by $S_1[2,4]$, $S_1[3,4]$,
and $\sqrt{a_1b_1}\ge |z_i|$ for $i=1,2,3$ by $S_1[1,2]$, $S_1[1,3]$. By the inequality $W_2[1,4]$
and the assumption $|z_4| > \sqrt{a_1b_1}$, we have
$$
\begin{aligned}
{1 \over 2}\lan W, \varrho \ran
& \ge \left( \sqrt{s_2t_2} + \sqrt{s_3t_3} + \sqrt{s_4t_4} -|u_4| \right) |z_4|
    + \left( \sqrt{s_1t_1} -|u_1| -|u_2| -|u_3|\right)\sqrt{a_1b_1} \\
& \ge \left( \sqrt{s_2t_2} + \sqrt{s_3t_3} + \sqrt{s_4t_4} -|u_4| \right) \sqrt{a_1b_1}
    + \left( \sqrt{s_1t_1} -|u_1| -|u_2| -|u_3|\right)\sqrt{a_1b_1} \\
& = \sqrt{a_1b_1}\left(\sum_{i=1}^4 \sqrt{s_i t_i} -\sum_{i=1}^4|u_i|\right).
\end{aligned}
$$
This is nonnegative by the inequality $W_3$, as it was desired.
\end{proof}

By the equivalence (ii) $\Longleftrightarrow$ (iii), we have the following criteria for the
convex hull of two basic dual cones:

\begin{theorem}\label{dual-join...}
For a self-adjoint $W=\xx(s,t,u)$ with nonnegative diagonals, we have the following:
\begin{enumerate}
\item[(i)]
$W\in\bbb^\circ+\ccc^\circ$ if and only if  $W_2[1,4]$, $W_2[2,3]$ and $W_3$ hold;
\item[(ii)]
$W\in\ccc^\circ+\aaa^\circ$ if and only if  $W_2[1,3]$, $W_2[2,4]$ and $W_3$ hold;
\item[(iii)]
$W\in\aaa^\circ+\bbb^\circ$ if and only if  $W_2[1,2]$, $W_2[3,4]$ and $W_3$ hold.
\end{enumerate}
If $W$ is a self-adjoint three qubit matrix with the {\sf X}-part $\xx(s,t,u)$, then the \lq only if\rq\ parts hold.
\end{theorem}

\begin{theorem}\label{ext-inter-states}
We have the following:
\begin{enumerate}
\item[(i)]
$\ext(\bbb\cap\ccc\cap\xx)=\E_{\bbb\cap\ccc}\cup \E_{\aaa\cap\bbb\cap\ccc}\cup\Delta$,\\
$\ext(\ccc\cap\aaa\cap\xx)=\E_{\ccc\cap\aaa}\cup \E_{\aaa\cap\bbb\cap\ccc}\cup\Delta$,\\
$\ext(\aaa\cap\bbb\cap\xx)=\E_{\aaa\cap\bbb}\cup \E_{\aaa\cap\bbb\cap\ccc}\cup\Delta$;
\item[(ii)]
$\ext((\bbb^\circ+\ccc^\circ)\cap\xx)=\ext({\bbb^\circ}\cap\xx) \cup \ext({\ccc^\circ}\cap\xx)$,\\
$\ext((\ccc^\circ+\aaa^\circ)\cap\xx)=\ext({\ccc^\circ}\cap\xx) \cup \ext({\aaa^\circ}\cap\xx)$,\\
$\ext((\aaa^\circ+\bbb^\circ)\cap\xx)=\ext({\aaa^\circ}\cap\xx) \cup \ext({\bbb^\circ}\cap\xx)$.
\end{enumerate}
\end{theorem}

\begin{proof}
(i). We will prove the first identity.
Suppose that $\varrho=\xx(a,b,z)$ satisfies $S_2^e[1,4]$ and
$$
\varrho = \xx(a',b',z')+\xx(a'',b'',z'') \quad \text{in} ~\bbb \cap \ccc \cap \xx.
$$
The condition $a_4=b_4=0$ implies
$a_4'=b_4'=z_j'=0=a_4''=b_4''=z_j''$ for $j=2,3,4$ by $S_1[2,4]$, $S_1[3,4]$.
Applying Lemma \ref{CS-S} with $(i,j)=(1,2)$ and $(1,3)$, we get
$(a',b',z') ~ \parallelsum (a'',b'',z'')$.

Next, suppose that $\varrho=\xx(a,b,z)$ satisfies $S_3^e$ and
$$
\varrho = \xx(a',b',z')+\xx(a'',b'',z'') \quad \text{in} ~\bbb \cap \ccc \cap \xx.
$$
Applying Lemma \ref{CS-S} with $(i,j)=(1,2)$,$(2,1)$,$(3,1)$, $(4,2)$, we get
$(a',b',z') ~ \parallelsum (a'',b'',z'')$.

(ii). States in $\ext({\bbb^\circ}\cap\xx) \cup \ext({\ccc^\circ}\cap\xx)$ generate extreme rays of the convex cone
$(\aaa^\circ+\bbb^\circ+\ccc^\circ) \cap \xx$ by Theorem \ref{ppt-ext}. Therefore, they
also generate extreme rays in the smaller cone $(\bbb^\circ+\ccc^\circ) \cap \xx$.
\end{proof}

Now, we look for extreme rays of $\bbb^\circ\cap\ccc^\circ\cap\xx$ (respectively, $\ccc^\circ\cap\aaa^\circ\cap\xx$ and $\aaa^\circ\cap\bbb^\circ\cap\xx$) to get conditions for the cone $\bbb+\ccc$ (respectively, $\ccc+\aaa$ and $\aaa+\bbb$).
To do this, we consider the condition
\begin{center}
\framebox{
\parbox[t][0.7cm]{14.00cm}{
\addvspace{0.1cm} \centering
$W^e_2[i,j]:\qquad \sqrt{s_it_i}=\sqrt{s_jt_j}=|u_k|=|u_\ell|=1$,\quad the others are zero,
}}
\end{center}\medskip
for $i\neq j$, where $k$ and $\ell$ are chosen so that $i,j,k,\ell$ are mutually distinct,
and define
$$
\begin{aligned}
\E_{\bbb^\circ\cap\ccc^\circ}&=\{W=\xx(s,t,u): W^e_2[i,j]\ {\text{\rm holds for}}\ (i,j)=(1,4)\ {\text{\rm or}}\  (2,3)\},\\
\E_{\ccc^\circ\cap\aaa^\circ}&=\{W=\xx(s,t,u): W^e_2[i,j]\ {\text{\rm holds for}}\ (i,j)=(1,3)\ {\text{\rm or}}\  (2,4)\},\\
\E_{\aaa^\circ\cap\bbb^\circ}&=\{W=\xx(s,t,u): W^e_2[i,j]\ {\text{\rm holds for}}\ (i,j)=(1,2)\ {\text{\rm or}}\  (3,4)\}.
\end{aligned}
$$
We also consider the following inequalities
\begin{center}
\framebox{
\parbox[t][0.7cm]{14.80cm}{
\addvspace{0.1cm} \centering
$S_2[i,j]:\qquad
\min\left\{\sqrt{a_ib_i}+\sqrt{a_jb_j},\sqrt{a_k b_k}+\sqrt{a_\ell b_\ell}\right\}
   \ge\max\left\{|z_i|+|z_j|,|z_k|+|z_\ell|\right\}$,
}}
\end{center}\medskip
for $i\neq j$, where $k,\ell$ are chosen so that $i,j,k,\ell$ are mutually distinct.

These inequalities have been used in \cite{han_kye_bisep_exam} to get necessary conditions
for a three state $\varrho$ with the {\sf X}-part $\xx(a,b,z)$ to
belong to $\bbb+\ccc$, $\ccc+\aaa$ and $\aaa+\bbb$ respectively. We show in Theorem \ref{join-state}
that they provide actually sufficient conditions when $\varrho$ itself {\sf X}-shaped.
Note that
\begin{itemize}
\item
$S_2[1,4]$ holds if and only if $S_2[2,3]$ holds;
\item
$S_2[1,3]$ holds if and only if $S_2[2,4]$ holds;
\item
$S_2[1,2]$ holds if and only if $S_2[3,4]$ holds.
\end{itemize}

\begin{lemma}\label{lemma-ntersection-witness}
For a given self-adjoint {\sf X}-shaped matrix $\varrho=\xx(a,b,z)$, the following are equivalent.
\begin{enumerate}
\item[(i)]
$\lan W,\varrho\ran\ge 0$ for each $W\in \E_{\bbb^\circ\cap\ccc^\circ}\cup\W^\Delta\cup\Delta$
{\rm (}respectively, $\E_{\ccc^\circ\cap\aaa^\circ}\cup\W^\Delta\cup\Delta$ and
$\E_{\aaa^\circ\cap\bbb^\circ}\cup\W^\Delta\cup\Delta${\rm )};
\item[(ii)]
$\varrho$ is a state satisfying the inequalities $S_2[1,4]$
{\rm (}respectively, $S_2[1,3]$ and $S_2[1,2]${\rm )} hold;
\item[(iii)]
$\lan W,\varrho\ran\ge 0$ for each $W\in\bbb^\circ\cap\ccc^\circ$
{\rm (}respectively, $\ccc^\circ\cap\aaa^\circ$ and $\aaa^\circ\cap\bbb^\circ${\rm )}.
\end{enumerate}
\end{lemma}

\begin{proof}
Although the proof of the direction (i) $\Longrightarrow$ (ii) already appears in \cite{han_kye_bisep_exam},
we include it here for the completeness.
We consider \xx-shaped three qubit self-adjoint matrices
$$
\begin{aligned}
W:=\xx\left(\left(\sqrt{b_1 \over a_1},0,0,\sqrt{b_4 \over a_4}\right),\left(\sqrt{a_1 \over b_1},0,0,\sqrt{a_4 \over b_4}\right),\left(0,-e^{-{\rm i}\theta_2},-e^{-{\rm i}\theta_3},0\right)\right) \\
W':=\xx\left(\left(0,\sqrt{b_2 \over a_2},\sqrt{b_3 \over a_3},0\right),\left(0,\sqrt{a_2 \over b_2},\sqrt{a_3 \over b_3},0\right), \left(-e^{-{\rm i}\theta_1},0,0,-e^{-{\rm i}\theta_4}\right)\right),
\end{aligned}
$$
for $\theta_i = \arg z_i$.
Then, both $W$ and $W'$ belong to $\E_{\bbb^\circ\cap\ccc^\circ}$.
We have
$$
\begin{aligned}
0 &\le {1 \over 2} \lan \varrho, W \ran = \sqrt{a_1b_1}+\sqrt{a_4b_4}-|z_2|-|z_3|,\\
0 &\le {1 \over 2} \lan \varrho, W' \ran = \sqrt{a_2b_2}+\sqrt{a_3b_3}-|z_1|-|z_4|.
\end{aligned}
$$

For the implication (ii) $\Longrightarrow$ (iii), suppose that $\varrho=\xx(a,b,z)$ satisfies $S_2[1,4]$.
By Corollary \ref{xxx-coro}, it suffices to show $\lan W, \varrho \ran \ge 0$ for every
$W=\xx(s,t,u)\in\bbb^\circ\cap\ccc^\circ$.
This is trivial when $\xx(s,t,u)$ is positive, that is, $\sqrt{s_i t_i} \ge |u_i|$ for all $i=1,2,3,4$.
We may assume without loss of generality that
\begin{equation}\label{vvvv}
0>\sqrt{s_1t_1}-|u_1| = \min \{ \sqrt{s_it_i}-|u_i| : i=1,2,3,4 \}.
\end{equation}
Then we have
$$
\begin{aligned}
\left(\sqrt{s_1t_1}+\sqrt{s_2t_2}\right)\sqrt{a_2b_2} &\ge \left(|u_1|+|u_2|\right)\sqrt{a_2b_2},\\
\left(\sqrt{s_1t_1}+\sqrt{s_3t_3}\right)\sqrt{a_3b_3} &\ge \left(|u_1|+|u_3|\right)\sqrt{a_3b_3},\\
\left(|u_1|+\sqrt{s_4t_4}\right)|z_4| &\ge \left(\sqrt{s_1t_1}+|u_4|\right)|z_4|,
\end{aligned}
$$
where the first and second inequalities follow from $W_1[1,2]$ and $W_1[1,3]$, respectively, and the last one
comes out from the equality in (\ref{vvvv}). Put $M=\sqrt{a_2b_2}+\sqrt{a_3b_3}-|z_4|$. Summing up the above three inequalities, we have
$$
\sqrt{s_1t_1}M
+\sum_{i=2}^3\sqrt{s_it_i}\sqrt{a_ib_i}+\sqrt{s_4t_4}|z_4|
\ge
|u_1|M + \sum_{i=2}^3|u_i|\sqrt{a_ib_i}+|u_4||z_4|,
$$
which implies
$$
\sqrt{s_1t_1}M
+\sum_{i=2}^4\sqrt{s_it_i}\sqrt{a_ib_i}
\ge |u_1|M
+\sum_{i=2}^4|u_i||z_i|.
$$
Because $M\ge |z_1|$ by $S_2[1,4]$ and $|u_1|-\sqrt{s_1t_1}>0$ by (\ref{vvvv}), we have
$$
\begin{aligned}
\sum_{i=2}^4(\sqrt{s_it_i}\sqrt{a_ib_i}-|u_i||z_i|)
&\ge (|u_1|-\sqrt{s_1t_1})M\\
&\ge (|u_1|-\sqrt{s_1t_1})|z_1|
\ge |u_1||z_1|-\sqrt{s_1t_1}\sqrt{a_1b_1}.
\end{aligned}
$$
This gives $\sum_{i=1}^4\sqrt{s_it_i}\sqrt{a_ib_i}\ge\sum_{i=1}^4|u_i||z_i|$, and $\lan W,\varrho\ran\ge 0$ by  by (\ref{ineq--xx}).
\end{proof}

Because $\bbb+\ccc$ (respectively, $\ccc+\aaa$ and $\aaa+\bbb$) is the dual of $\bbb^\circ\cap\ccc^\circ$ (respectively, $\ccc^\circ\cap\aaa^\circ$ and $\aaa^\circ\cap\bbb^\circ$), the equivalence between (ii) and (iii) of
Lemma \ref{lemma-ntersection-witness} gives rise to the following characterization
of the cone $\bbb+\ccc$ (respectively, $\ccc+\aaa$ and $\aaa+\bbb$) for {\sf X}-states.

\begin{theorem}\label{join-state}
For a three qubit {\sf X}-state $\varrho=\xx(a,b,z)$, we have the following:
\begin{enumerate}
\item[(i)]
$\varrho\in\bbb+\ccc$ if and only if $S_2[1,4]$ holds;
\item[(ii)]
$\varrho\in\ccc+\aaa$ if and only if $S_2[1,3]$ holds;
\item[(iii)]
$\varrho\in\aaa+\bbb$ if and only if $S_2[1,2]$ holds.
\end{enumerate}
For a general three qubit state $\varrho$ with the {\sf X}-part $\xx(a,b,z)$, the \lq only if\rq\ parts hold.
\end{theorem}

\begin{theorem}\label{ext-inter-wit}
We have the following:
\begin{enumerate}
\item[(i)]
$\ext(\bbb^\circ\cap\ccc^\circ\cap\xx)=\E_{\bbb^\circ\cap\ccc^\circ}\cup\W^\Delta\cup\Delta$,\\
$\ext(\ccc^\circ\cap\aaa^\circ\cap\xx)=\E_{\ccc^\circ\cap\aaa^\circ}\cup\W^\Delta\cup\Delta$,\\
$\ext(\aaa^\circ\cap\bbb^\circ\cap\xx)=\E_{\aaa^\circ\cap\bbb^\circ}\cup\W^\Delta\cup\Delta$;
\item[(ii)]
$\ext((\bbb+\ccc)\cap\xx)=\ext(\bbb\cap\xx)\cup\ext(\ccc\cap\xx)$,\\
$\ext((\ccc+\aaa)\cap\xx)=\ext(\ccc\cap\xx)\cup\ext(\aaa\cap\xx)$,\\
$\ext((\aaa+\bbb)\cap\xx)=\ext(\aaa\cap\xx)\cup\ext(\bbb\cap\xx)$.
\end{enumerate}
\end{theorem}

\begin{proof}
(i). We will prove the first identity.
Since elements in $\Delta$ and $\W^\Delta$ are extremal in $\bbb^\circ \cap \xx$, they are also extremal in the smaller cone
$\bbb^\circ \cap \ccc^\circ \cap \xx$.
Suppose that $W=\xx(s,t,u)$ satisfies $W_2^e[1,4]$ and
$$
W = \xx(s',t',u')+\xx(s'',t'',u'') \quad \text{in} ~\bbb^\circ \cap \ccc^\circ \cap \xx.
$$
Applying Lemma \ref{CS-W} with $(i,j,k,\ell)=(1,2,2,1)$,$(4,2,2,4)$,$(1,3,3,1)$,$(4,3,3,4)$, we get
$(s',t',u') ~ \parallelsum (s'',t'',u'')$.

(ii). Since states in $\ext(\bbb\cap\xx)\cup\ext(\ccc\cap\xx)$ are extremal in
the convex cone $(\aaa+\bbb+\ccc) \cap \xx$
by Theorem \ref{bi-sep-ext}, they are also extremal in the smaller cone $(\bbb+\ccc) \cap \xx$.
\end{proof}

\section{Summary}

In this paper, we have considered the convex cones in the diagrams (\ref{diagram}) and (\ref{diagram1})
arising from classification of partial separability/entanglement of three qubit states and their witnesses. For those convex cones,
we got the following results:
\begin{itemize}
\item
characterization for {\sf X}-shaped matrices by algebraic inequalities, which give rise to
necessary criteria for general three qubit states/witnesses in terms of diagonal and anti-diagonal entries;
\item
finding all the extreme rays of the cones consisting of {\sf X}-shaped matrices, with which
we may exhibit all {\sf X}-shaped matrices in the cones.
\end{itemize}

\begin{table}
\begin{center}
\begin{tabular}{|c|c|c||c|c|c|}
\hline
 criteria & extreme &states & witnesses& extreme & criteria\\
 \hline
 \hline
 $S_1[1,4]$, &  $S^e_1[1,4]$  &       &               & $W^e_1[1,4]$  &  $W_1[1,4]$ \\
 &    & &   & $W^e_1[4,1]$  &   \\
 $S_1[2,3]$ & $S^e_1[2,3]$ & $\aaa\cap\xx$ & $\aaa^\circ\cap\xx$ & $W^e_1[2,3]$  & $W_1[2,3]$ \\
 &    &       &               & $W^e_1[3,2]$  &   \\
 Prop.\,\ref{aa-bbb-ccc} &Th.\,\ref{basic_dual} && & Th.\,\ref{basic-dual-ext} & Prop.\,\ref{basic-dual-cri} \\
 \hline
 $S_1[i,j]$    & $S_3^e$      &$\aaa\cap\bbb\cap\ccc\cap\xx$ & $(\aaa^\circ+\bbb^\circ+\ccc^\circ)\cap\xx$ & $W_1^e[i,j]$ & $W_3$ \\
 Prop.\,\ref{aa-bbb-ccc} & Th.\,\ref{ppt-ext}(i)&&& Th.\,\ref{ppt-ext}(ii) & Prop.\,\ref{coro_decomp}\\
 \hline
 $S_3$      & $S^e_1[i,j]$   &$(\aaa+\bbb+\ccc) \cap \xx$ & $\aaa^\circ\cap\bbb^\circ\cap\ccc^\circ\cap\xx$ & $W^e_3[i]$ & $W_1[i,j]$ \\
 Prop.\,\ref{i_join}&Th.\,\ref{bi-sep-ext}(i)&&& Th.\,\ref{bi-sep-ext}(ii) & Prop.\,\ref{basic-dual-cri} \\
 \hline
 $S_1[1,4]$    &$S^e_3$   &               &                             & $W_1^e[1,4]$ & $W_3$ \\
  &  &               &                             & $W_1^e[4,1]$ &  \\
 $S_1[2,3]$    &$S^e_2[1,2]$   &               &                             & $W_1^e[2,3]$ & $W_2[1,2]$ \\
 & $S^e_2[2,1]$ &               &                             & $W_1^e[3,2]$ & \\
 $S_1[1,3]$    &   & $\aaa\cap\bbb\cap\xx$ & $(\aaa^\circ+\bbb^\circ)\cap\xx$ & $W_1^e[1,3]$ &  \\
 & $S^e_2[3,4]$ & & & $W_1^e[3,1]$ & $W_2[3,4]$ \\
 $S_1[2,4]$    &$S^e_2[4,3]$   &               &                             & $W_1^e[2,4]$ & \\
 & &               &                             & $W_1^e[4,2]$ & \\
 Prop.\,\ref{aa-bbb-ccc} &Th.\,\ref{ext-inter-states}(i) &&& Th.\,\ref{ext-inter-states}(ii) & Th.\,\ref{dual-join...} \\
 \hline
               &$S^e_1[1,4]$    &                 &                             & & $W_1[1,4]$       \\
$S_2[1,2]$     &$S^e_1[2,3]$    &                 &                             & $W^e_2[1,2]$ & $W_1[2,3]$ \\
               &$S^e_1[1,3]$    & $(\aaa+\bbb)\cap\xx$ & $\aaa^\circ\cap\bbb^\circ\cap\xx$ & $W^e_2[3,4]$ & $W_1[1,3]$ \\
               &$S^e_1[2,4]$    &                 &                             & & $W_1[2,4]$ \\
Th.\,\ref{join-state}&Th.\,\ref{ext-inter-wit}(ii)&&&Th.\,\ref{ext-inter-wit}(i) & Prop.\,\ref{basic-dual-cri} \\
\hline
\end{tabular}
\caption{Criteria and extreme rays of convex cones: Conditions for
\lq criteria\rq\ and \lq extreme\rq\ are connected by \lq {\sf
and}\rq\ and \lq {\sf or}\rq, respectively.
}
\end{center}
 \end{table}

We summarize the results in Table 1. We note our characterizion is one of very few cases when we may check separability
by inequalities, without decomposing into the sum of pure product states.
For example, we may check separability for $2\otimes 2$ and $2\otimes 3$ cases
by the PPT condition. We may also check full separability of multi-qubit {\sf X}-states by
inequalities \cite{{chen_han_kye},{ha-han-kye},{han_kye_GHZ}}. Checking separability with inequalities in this paper
was possible through the duality and characterizing extreme rays of the dual cones.

This work has been partly motivated by the questions \cite{sz2012}
on the existence of states in the seven classes arising in the
classification of partial entanglement, including the following
classes:
$$
\begin{aligned}
\C^{2,6,1} &:= \aaa \cap (\bbb + \ccc)\cap \bbb^{\rm c} \cap \ccc^{\rm c},\\
\C^{2,4} &:= (\aaa + \bbb) \cap (\bbb + \ccc) \cap (\ccc + \aaa)\cap \aaa^{\rm c} \cap \bbb^{\rm c} \cap \ccc^{\rm c},\\
\C^{2,3,1} &:= (\aaa + \bbb) \cap (\ccc + \aaa) \cap \aaa^{\rm c} \cap (\bbb + \ccc)^{\rm c},
\end{aligned}
$$
together with the convex cones obtained by permuting systems.
Here, $\C^{2,6,1}$, $\C^{2,4}$ and $\C^{2,3,1}$ are notations in \cite{sz2012}.
The authors \cite{han_kye_bisep_exam} gave examples of {\sf X}-shaped states belonging to those classes.
In this paper, we gave complete necessary and sufficient conditions for {\sf X}-states to be members of the classes.
For example, an {\sf X}-state $\varrho=\xx(a,b,z)$ belongs to the class $\C^{2,6,1}$ if and only if the following
hold:
\begin{itemize}
\item
$\varrho$ satisfies the inequalities $S_1[1,4]$, $S_1[2,3]$ and $S_2[1,4]$;
\item
$\varrho$ violates $S_1[1,3]$ or $S_1[2,4]$;
\item
$\varrho$ violates $S_1[1,2]$ or $S_1[3,4]$.
\end{itemize}
The example $\varrho=\xx((0,1,1,2),(0,1,1,2),(0,1,1,0))$ given in \cite{han_kye_bisep_exam} satisfies
$S_1[1,4]$, $S_1[2,3]$ and $S_2[1,4]$, but violates $S_1[1,3]$ and $S_1[1,2]$.

It is natural to ask what happens in the four qubit system, or arbitrary qubit systems.
We began with the result \cite{han_kye_optimal} that
an {\sf X}-shaped multi-qubit state is separable with respect to a {\sl bi-partition}
of systems if and only if it is of positive partial transpose
with respect to the same {\sl bi-partition}. This was crucial to give characterizations
in terms of diagonal entries and the {\sl modulus} of anti-diagonal entries. But this is not the case
for tri-partitions.  In the three qubit system,
considering tri-partition is amount to full separability. We need the phase parts, that is, the angular parts of anti-diagonal entries, as well as the modulus parts to characterize full separability of three qubit \xx-states.
See \cite{{chen_han_kye},{ha-han-kye},{han_kye_phase}}.
We note that all kinds of partial separability come out from bi-partitions in the three qubit case.
But, it is necessary to consider tri-partitions as well as bi-partitions in the four qubit case.
See \cite{sz2015,sz2018}. Therefore, exploring partial separability/entanglement
in general qubit system must be a very challenging project even for {\sf X}-shaped states.


\end{document}